\documentclass[11pt]{article}
\usepackage[T1]{fontenc}
\usepackage[utf8]{inputenc}
\usepackage{lmodern}
\usepackage{microtype}
\usepackage[a4paper,margin=1in]{geometry}
\usepackage{amsmath,amssymb,amsthm,mathtools}
\usepackage{booktabs,tabularx}
\usepackage{aliascnt}
\usepackage{xcolor}
\usepackage{tikz}
\usetikzlibrary{arrows.meta,positioning,decorations.pathreplacing,calc}
\usepackage{hyperref}
\usepackage[capitalise,nameinlink,noabbrev]{cleveref}

\allowdisplaybreaks[2]

\hypersetup{
  colorlinks=true,
  linkcolor=blue!55!black,
  citecolor=green!45!black,
  urlcolor=blue!60!black
}

\theoremstyle{plain}
\newtheorem{theorem}{Theorem}[section]
\newaliascnt{lemma}{theorem}
\newtheorem{lemma}[lemma]{Lemma}
\aliascntresetthe{lemma}
\newaliascnt{proposition}{theorem}
\newtheorem{proposition}[proposition]{Proposition}
\aliascntresetthe{proposition}
\newaliascnt{corollary}{theorem}
\newtheorem{corollary}[corollary]{Corollary}
\aliascntresetthe{corollary}

\theoremstyle{definition}
\newaliascnt{definition}{theorem}
\newtheorem{definition}[definition]{Definition}
\aliascntresetthe{definition}
\newaliascnt{example}{theorem}

\aliascntresetthe{example}

\theoremstyle{remark}
\newaliascnt{remark}{theorem}
\newtheorem{remark}[remark]{Remark}
\aliascntresetthe{remark}

\usepackage{thm-restate}

\crefname{lemma}{Lemma}{Lemmas}
\Crefname{lemma}{Lemma}{Lemmas}
\crefname{proposition}{Proposition}{Propositions}
\Crefname{proposition}{Proposition}{Propositions}
\crefname{corollary}{Corollary}{Corollaries}
\Crefname{corollary}{Corollary}{Corollaries}
\crefname{definition}{Definition}{Definitions}
\Crefname{definition}{Definition}{Definitions}
\crefname{example}{Example}{Examples}
\Crefname{example}{Example}{Examples}
\crefname{remark}{Remark}{Remarks}
\Crefname{remark}{Remark}{Remarks}

\newcommand{\QAC}{\mathsf{QAC}}
\newcommand{\QACzero}{\mathsf{QAC}^{0}}
\newcommand{\twoDQAC}{\mathsf{2D\text{-}QAC}}
\newcommand{\twoDQACzero}{\mathsf{2D\text{-}QAC}^{0}}
\newcommand{\oneDQAC}{\mathsf{1D\text{-}QAC}}
\newcommand{\oneDQACzero}{\mathsf{1D\text{-}QAC}^{0}}
\newcommand{\Parity}{\mathsf{PARITY}}
\newcommand{\CZ}{\mathsf{CZ}}

\newcommand{\bra}[1]{\left\langle #1\right|}
\newcommand{\ket}[1]{\left|#1\right\rangle}
\newcommand{\braket}[2]{\left\langle #1\,\middle|\,#2\right\rangle}
\newcommand{\ketbra}[2]{\left|#1\right\rangle\!\left\langle #2\right|}
\newcommand{\proj}[1]{\ketbra{#1}{#1}}
\newcommand{\Cat}[1]{\ket{\mathrm{cat}_{#1}}}
\newcommand{\Tr}{\operatorname{Tr}}

\newcommand{\trdist}[2]{D_{\mathrm{tr}}\!\left(#1,#2\right)}
\newcommand{\fidelity}[2]{F\!\left(#1,#2\right)}

\newcommand{\LocalMeasurementVariance}[1]{\Gamma\!\left(#1\right)}

\newcommand{\AverageSensitivity}{\operatorname{\mathbf{as}}}
\newcommand{\FourierWeight}[2]{\mathbf{W}^{#1}\!\left[#2\right]}
\newcommand{\E}{\mathbb{E}}
\newcommand{\br}[1]{\left( #1 \right)}
\newcommand{\Br}[1]{\left[ #1 \right]}
\newcommand{\set}[1]{\left\{ #1 \right\}}
\newcommand{\bits}{\{0,1\}}

\newcommand{\Var}{\operatorname{Var}}

\newcommand{\abs}[1]{\left|#1\right|}
\newcommand{\norm}[1]{\left\lVert #1\right\rVert}
\newcommand{\ceil}[1]{\left\lceil #1\right\rceil}
\newcommand{\floor}[1]{\left\lfloor #1\right\rfloor}
\newcommand{\supp}{\operatorname{supp}}
\newcommand{\poly}{\operatorname{poly}}

\DeclareMathOperator{\diag}{diag}
\usepackage{quantikz}

\title{On the Computational Power of Geometrically Local $\mathrm{QAC}$ circuits}
\author{
 Yangjing Dong\thanks{\scriptsize  State Key Laboratory of Novel Software Technology, Nanjing University, Nanjing 210023, China. Email: dongmassimo@gmail.com.}~\thanks{\scriptsize Aalto University, Finland.}
\and Fengning Ou\thanks{\scriptsize State Key Laboratory of Novel Software Technology, Nanjing University, Nanjing 210023, China. Email: reverymoon@gmail.com.}
\and Penghui Yao\thanks{\scriptsize  State Key Laboratory of Novel Software Technology, Nanjing University, Nanjing 210023, China.  Email: phyao1985@gmail.com.}~\thanks{\scriptsize Hefei National Laboratory, Hefei 230088, China.}
}

\begin{document}
\maketitle
\begin{abstract}
The computational power of $\mathsf{QAC}$, which is a constant-depth quantum circuit family consisting of arbitrary single-qubit unitaries and $n$-qubit generalized Toffoli gates, has gained tremendous focus recently. In this work, we initiate the study of the computational complexity of geometrically local $\mathsf{QAC}$ circuits, where the qubits are arranged in a high-dimensional grid, and all the generalized Toffoli gates act on contiguous qubits along one direction. We prove that $\mathsf{2D\text{-}QAC}$ circuits are universal: any $\mathsf{QAC}$ circuit can be exactly simulated by a two-dimensional geometrically local $\mathsf{QAC}$ circuit. Therefore, we focus on the computational power of $\mathsf{2D\text{-}QAC}$ circuits with grid width $w$. We prove that $\mathsf{PARITY}$ computation has a tight width--depth tradeoff of $w=n^{\Theta(1/d)}$ in the constant-depth regime if the circuit allows arbitrarily many ancillas.
 We also establish Fourier tail and sensitivity bounds, together with tight width lower bounds in $\mathsf{2D\text{-}QAC}$ with fixed depth and tight depth lower bounds in $\mathsf{1D\text{-}QAC}$ circuits for  $\mathsf{MAJORITY}$ and $\mathsf{MOD}$ computation.

Our main tool is \emph{local measurement variance}, the maximum variance of normalized sums of observables on separated regions. We use a hybrid argument to show that each layer of a $\mathsf{2D\text{-}QAC}$ circuit increases this variance by at most a factor depending only on its width, yielding a variance bound for every $\mathsf{2D\text{-}QAC}$ state. The high local measurement variance of cat and Dicke states then gives the width lower bounds.
\end{abstract}

\newpage
\tableofcontents
\clearpage
\section{Introduction}
\label{sec:introduction}

Constant-depth geometrically local quantum circuits have demonstrated provable computational
advantages for relation and sampling problems, including with geometrically local gates
\cite{doi:10.1126/science.aar3106,10.1145/3313276.3316404,10.1145/3357713.3384332,watts2024unconditionalquantumadvantagesampling,grier2025quantumadvantagesamplingshallow}.
For single-output Boolean computation, however, bounded-arity circuits face a basic obstruction: the output’s backward light cone contains only constantly many input bits. Such circuits therefore cannot compute functions that depend on all bits of a growing input.

To study quantum computation beyond this light-cone restriction,
Moore introduced $\QACzero$, the class of constant-depth,
polynomial-size quantum circuits with arbitrary one-qubit unitaries
and generalized Toffoli gates of unbounded arity
\cite{moore1999quantum}. 
Recent work has established quantum advantages of $\QACzero$ for
Boolean computation. Grier, Morris, and Wu exhibited total Boolean
functions computable by $\QACzero$ circuits but requiring
exponential-size $\mathsf{AC}^0[p]$ circuits for every fixed prime $p$
\cite{grier2026mathsfqac0containsmathsftc0with}.
Their subsequent work extended this separation to exact computation
\cite{grier2026robustness}.
Gretta, Gupta, and Joshi obtained an average-case separation from
$\mathsf{AC}^0$ using MAJORITY \cite{GGJ26}.

In this paper, we investigate the computational power of $\QAC$ circuits, where the generalized Toffoli gates may act on an arbitrary subset of qubits and the circuits allow an arbitrary number of ancillas which are initialized to product states. We initiate the study of {\em geometrically local} $\QAC$ circuits, in which qubits occupy a high-dimensional grid and every generalized Toffoli gate acts on a contiguous interval along one direction. Geometric locality provides a natural way to restrict this power while still allowing unbounded-arity interactions. Such restrictions are motivated both by spatially constrained
architectures~\cite{hendrickx2021four,kim2023utility,WillowSpec} and by the broader question of which forms of locality limit shallow quantum
computation~\cite{doi:10.1126/science.aar3106,10.1145/3313276.3316404}.  Unlike bounded-arity locality, however, geometric locality does not immediately impose a small light cone, since a single gate may still span an arbitrarily long region. This leads us to ask, 
\begin{center}
What is the computational power of geometrically local $\QAC$ circuits?    
\end{center}

We first show that  $\twoDQAC$ circuits are univeral: circuits whose multi-qubit gates act
on intervals in a single row or column exactly simulate general
$\QAC$ circuits with constant-factor depth overhead and quadratic
qubit overhead.  As an immediate corollary,  $\twoDQACzero=\QACzero$.
Thus, imposing this two-dimensional geometry alone does not reduce computational power. The relevant restriction is the {\em width} of the layout. Accordingly, the central focus of this paper is to understand how much width is necessary and sufficient for shallow computation, while allowing the length of the layout---and hence the number of product ancillas--to be arbitrarily large.

For exact computation of PARITY, MAJORITY, and MOD, we show that
depth-$d$ $\twoDQAC$ circuits require width $n^{\Omega(1/d)}$. The same lower bound holds for approximate
cat-state synthesis. Our exact clean constructions achieve a matching width
$n^{O(1/d)}$ at depth $O(d)$ using exponentially many ancillas.

To study the computational power of $\twoDQAC$ circuits, we introduce a novel notion, local measurement variance, as a potential that captures correlations among the qubits generated by the circuit. We show that the initial product state has low local measurement variance, while cat states have large local measurement variance. Moreover, we prove an upper bound on the growth of local measurement variance through each layer. 
For the upper bounds, we adapt
Rosenthal's synthesis and the exact amplification of Grier, Morris,
and Wu to the grid
\cite{rosenthal:LIPIcs.ITCS.2021.32,grier2026mathsfqac0containsmathsftc0with},
and combine small circuits recursively to reduce the width.

The one-dimensional setting is a particularly interesting extreme of this
width-based perspective: here the geometry is maximally restrictive,
yet unbounded-arity gates can still span arbitrarily long intervals.
A backward-observable argument gives linear depth lower bounds for
approximate PARITY, MAJORITY, and fixed-modulus MOD computation.
Fot the lower bounds, we allow the inputs to be placed anywhere and allow arbitrarily many ancillas. And we also give exact clean constructions with matching linear
depth for these Boolean functions.

\subsection{Main results}
\label{sec:main-results}

Our first result shows that two-dimensional $\QAC$ circuits are univeral.
The $\twoDQAC$ model places qubits on a rectangular grid of $w$ rows
and $h$ columns, called its width and length. It allows arbitrary
one-qubit unitaries and generalized Toffoli gates of unbounded arity,
but each multi-qubit gate must act on a contiguous interval in a single
row or column. Gates in the same layer have disjoint supports, as
illustrated in \cref{fig:geometric-qac-layer}. Depth counts these
multi-qubit layers; one-qubit gates may be applied between them.

\begin{figure}[htbp]
\centering
\begin{tikzpicture}[
  x=0.85cm,y=0.85cm,
  wire/.style={draw=black!85,line width=0.9pt},
  control/.style={circle,fill=black!85,inner sep=0pt,minimum size=4.5pt},
  target/.style={circle,draw=black!85,fill=white,line width=0.9pt,
    inner sep=0pt,minimum size=9pt,
    path picture={
      \draw[black!85,line width=0.9pt]
        (path picture bounding box.north)--(path picture bounding box.south)
        (path picture bounding box.west)--(path picture bounding box.east);
    }}
]
  \fill[blue!20] (0,1) rectangle (1,4);
  \fill[blue!20] (6,0) rectangle (7,4);
  \fill[green!20] (0,0) rectangle (4,1);
  \fill[green!20] (2,3) rectangle (4,4);
  \fill[green!20] (2,2) rectangle (6,3);
  \fill[green!20] (1,1) rectangle (5,2);
  \draw[step=1,black!22,line width=0.35pt] (0,0) grid (8,4);
  \draw[black!45,line width=0.6pt] (0,0) rectangle (8,4);

  \draw[wire] (0.5,3.5)--(0.5,1.5);
  \foreach \y in {3.5,2.5} \node[control] at (0.5,\y) {};
  \node[target] at (0.5,1.5) {};

  \draw[wire] (0.5,0.5)--(3.5,0.5);
  \foreach \x in {0.5,1.5,2.5} \node[control] at (\x,0.5) {};
  \node[target] at (3.5,0.5) {};

  \draw[wire] (2.5,3.5)--(3.5,3.5);
  \node[control] at (2.5,3.5) {};
  \node[target] at (3.5,3.5) {};

  \draw[wire] (2.5,2.5)--(5.5,2.5);
  \node[target] at (2.5,2.5) {};
  \foreach \x in {3.5,4.5,5.5} \node[control] at (\x,2.5) {};

  \draw[wire] (1.5,1.5)--(4.5,1.5);
  \foreach \x in {1.5,2.5,3.5} \node[control] at (\x,1.5) {};
  \node[target] at (4.5,1.5) {};

  \draw[wire] (6.5,3.5)--(6.5,0.5);
  \node[target] at (6.5,3.5) {};
  \foreach \y in {2.5,1.5,0.5} \node[control] at (6.5,\y) {};

  \draw[<->,>=Stealth,black,line width=0.8pt] (0,-0.38)--node[below,font=\normalsize,text=black]{$h=8$}(8,-0.38);
  \draw[<->,>=Stealth,black,line width=0.8pt] (-0.38,0)--node[left,font=\normalsize,text=black]{$w=4$}(-0.38,4);
\end{tikzpicture}
\caption{One multi-qubit layer of a $\twoDQAC$ circuit on a $4\times8$
lattice. Each cell represents a qubit. Green indicates horizontal gates,
and blue indicates vertical gates.}
\label{fig:geometric-qac-layer}
\end{figure}

The class $\twoDQACzero$ consists of languages that can be decided, i.e., computed with bounded error, by a family of constant-depth and polynomial-sized $\twoDQAC$ circuits. Width one gives $\oneDQACzero$.

\begin{theorem}[{Informal version of \cref{thm:two-dimensional-simulation}}]
\[
 \twoDQACzero=\QACzero.
\]
Every unrestricted $\QAC$ circuit has an exact $\twoDQAC$ simulation
with constant-factor depth overhead and quadratic qubit overhead.
\end{theorem}

This equivalence lets us study $\QACzero$ through the width of a
two-dimensional layout. We establish a connection between the width of $\twoDQAC$ circuits and the Fourier spectrum of the functions the circuits compute. More specifically, narrow grids constrain the Fourier spectrum
even when their length is unrestricted. Write $\FourierWeight{>k}{f}$
for the Fourier weight above degree $k$ and $\AverageSensitivity(f)$
for the average sensitivity.

\begin{theorem}[{Informal version of \cref{thm:fourier-tail,thm:boolean-sensitivity-bound}}]
For fixed $d$, suppose a depth-$d$, width-$w$ $\twoDQAC$ circuit
computes $f:\{-1,1\}^n\to\{-1,1\}$ exactly from product ancillas.
For every integer $k\ge1$,
\[
 \FourierWeight{>k}{f}
 \le O\left(\frac{n\cdot w^{6d+1}}{k^2}\right),
\]
and
\[
 \AverageSensitivity(f)
 \le O\left(\sqrt n\cdot w^{3d+1/2}\right).
\]
The implicit constants may depend on $d$.
\end{theorem}

These estimates and the underlying output variance bound give bounds on the width for Boolean computation and state synthesis, summarized in \cref{tab:geometric-bounds}.

\begin{table}[htbp]
\centering
\small
\renewcommand{\arraystretch}{1.24}
\begin{tabular}{@{}l l >{\raggedright\arraybackslash}p{0.34\linewidth} l@{}}
\toprule
Constraint & Task & Depth / width bound & Theorems \\
\midrule
$\oneDQAC$ & PARITY & $d=\Theta(n)$ &
\ref{thm:one-dimensional-parity-lower-bound}, \ref{thm:exact-parity-construction} \\
$\oneDQAC$ & MAJORITY & $d=\Theta(n)$ &
\ref{thm:one-dimensional-majority-mod-lower-bound}, \ref{thm:one-dimensional-boolean-synthesis} \\
$\oneDQAC$ & MOD & $d=\Theta(n)$ &
\ref{thm:one-dimensional-majority-mod-lower-bound}, \ref{thm:one-dimensional-boolean-synthesis} \\
$\oneDQAC$ & Cat state & $d=\Theta(\log n)$ &
\ref{thm:width-lower-bound}, \ref{thm:one-dimensional-cat-construction} \\
\addlinespace[0.4em]
$\twoDQAC$ & PARITY & $w=n^{\Theta(1/d)}$ &
\ref{thm:parity-sensitivity}, \ref{thm:exact-parity-construction} \\
$\twoDQAC$ & MAJORITY & $w=n^{\Theta(1/d)}$ &
\ref{thm:majority-sensitivity}, \ref{thm:boolean-synthesis} \\
$\twoDQAC$ & MOD & $w=n^{\Theta(1/d)}$ &
\ref{thm:parity-sensitivity}, \ref{thm:mod-sensitivity}, \ref{thm:boolean-synthesis} \\
$\twoDQAC$ & Cat state & $w=n^{\Theta(1/d)}$ &
\ref{thm:width-lower-bound}, \ref{thm:one-dimensional-cat-construction} \\
$\twoDQAC$ & Dicke state $\ket{D_{n,k}}$ &
$w=\Omega((1+\min\{k,n-k\})^{1/(3d+1)})$ &
\ref{thm:width-lower-bound} \\
\bottomrule
\end{tabular}
\caption{Depth bounds in one dimension and width bounds for fixed depth
$d$ in two dimensions. The matching two-dimensional constructions have
depth $O(d)$ and unrestricted ancillary workspace.}
\label{tab:geometric-bounds}
\end{table}
\newpage
\begin{theorem}[{Informal version of \cref{thm:width-lower-bound,thm:parity-sensitivity,thm:majority-sensitivity,thm:mod-sensitivity}}]
For depth-$d$ $\twoDQAC$ circuits with unbounded ancillas:
\begin{itemize}
\item Computing $\Parity$, $\mathsf{MAJORITY}$, or
$\mathsf{MOD}$ requires width $n^{\Omega(1/d)}$.

\item Synthesizing an $n$-qubit cat state approximately requires width $n^{\Omega(1/d)}$.
\item Exact synthesis of the Dicke state $D_{n,k}$ requires width
$(1+\min\{k,n-k\})^{\Omega(1/d)}$.
\end{itemize}
\end{theorem}

For PARITY, MAJORITY, MOD, and cat-state synthesis,
these lower bounds are matched with exponential ancillary workspace. When the number of ancillas is unrestricted, these circuits
unconditionally match our width lower bounds.

\begin{theorem}[{Informal version of \cref{thm:exact-parity-construction,thm:boolean-synthesis,thm:one-dimensional-cat-construction}}]
For fixed $d$, exponentially many ancillas suffice for exact clean
synthesis of $\Parity$, $\mathsf{MAJORITY}$, $\mathsf{MOD}$,
and cat states by depth-$O(d)$ $\twoDQAC$ circuits
of width $n^{\Theta(1/d)}$.
\end{theorem}

Using the similar ancilla-reduction argument of Anshu, Dong, Ou, and
Yao~\cite{ADOY24}, 
we also give a recursive construction that trades depth for width while preserving polynomial size, allowing the depth to grow with the input length.

\begin{theorem}[{Informal version of \cref{thm:width-recursion}}]
Suppose $\Parity$ can be computed with bounded error by polynomial-size $\QAC$
circuits of nondecreasing depth $d(n)\ge1$. Then it can be computed
by polynomial-size $\twoDQAC$ circuits of depth $O(d(n)^2)$ and width
$O(d(n)n^{O(1/d(n))})$.
\end{theorem}


As a special case of $\twoDQAC$,
we also study the computational power of $\oneDQAC$ and prove a linear bound on the depth.

\begin{theorem}[{Informal version of \cref{thm:width-lower-bound,thm:one-dimensional-parity-lower-bound,thm:one-dimensional-majority-mod-lower-bound}}]
For a $\oneDQAC$ circuit with unbounded ancillas:
\begin{itemize}
    \item Computing $\Parity$, $\mathsf{MAJORITY}$, or $\mathsf{MOD}$ requires depth $\Omega\br{n}$.
    \item Synthesizing an $n$-qubit cat state approximately requires depth $\Omega\br{\log n}$.
\end{itemize}
\end{theorem}
The above lower bounds are tight in the sense that,
with a constant number of ancillas,
the lower bounds can be matched by $\oneDQAC$ circuits.
\begin{theorem}[{Informal version of \cref{thm:exact-parity-construction,thm:one-dimensional-boolean-synthesis,thm:one-dimensional-cat-construction}}]
  The exist $\oneDQAC$ circuits such that
  \begin{itemize}
      \item Compute $\Parity$ with depth $O(n)$, and no ancillas.
      \item Compute $\mathsf{MAJORITY}$ with depth $O(n)$, and $O(\log n)$ ancillas.
      \item Compute $\mathsf{MOD}$ with depth $O(n)$, and $O(1)$ ancillas.
      \item Synthesize $n$-qubit cat state with depth $O(\log n)$, and no ancillas.
  \end{itemize}
\end{theorem}

\subsection{Technique Overview}
\label{sec:technique-overview}

\paragraph{Local measurement variance.}
Our lower bounds use the variance of local measurement outcomes to
distinguish cat states from product states. For a state $\rho$ on a grid, define
\[
 \LocalMeasurementVariance{\rho}
 =\max\Var_\rho\left(\sum_i c_iH_i\right),
\]
where the variance for a state $\rho$ and Hermitian observable $H$ is defined as
\[
 \Var_\rho(H)=\Tr(\rho H^2)-\bigl(\Tr(\rho H)\bigr)^2,
\]
and the maximum is taken over sets of Hermitian observables $\set{H_i}_{i=1}^m$ and real coefficients $\set{c_i}_{i=1}^m$ satisfying
\begin{itemize}
    \item each $H_i$ is a contraction, i.e., $\norm{H}\leq 1$, acting on continuous qubits in the grid;
    \item for any $i$ and $j$, $H_i$ and $H_j$ acts on disjoint columns in the grid;
    \item The real coefficients satisfy $\sum_{i=1}^mc_i^2=1$.
\end{itemize}

To see the intuition,
let us compare the variance of a product state and the variance of a cat state,
with respect to the Hermitian observable
$H=\sum_{i=1}^n Z_i/\sqrt{n}$.
Measuring $H$ amounts to measuring $Z$ on each qubit and summing the
outcomes. For a product state $\ket\rho=\bigotimes_{i=1}^n\ket{\rho_i}$,
each outcome $Y_i$ is $+1$ or $-1$, with probabilities determined by
$\ket{\rho_i}$. The outcomes are independent, so their variances sum up as:
\[
 \Var_\rho(H)
 =\frac{1}{n}\sum_{i=1}^n\Var(Y_i)
 =\frac{1}{n}\sum_{i=1}^n\left(1-(\E Y_i)^2\right)\le 1.
\]
On the contrary, for the cat state $\Cat n$, the outcomes are either all $+1$ or all
$-1$, each with probability $1/2$. Thus the same observable $H$ takes
the values $+n$ and $-n$ with equal probability. Its variance is
\[
 \Var_{ \mathrm{cat}_n }(H)=n.
\]
The variance grows linearly with the number of qubits for cat
states, but is upper bounded by a constant for product states. Local measurement
variance captures this difference using normalized observables on
separate regions.
Let $U$ be a layer of a $\twoDQAC$ circuit,
we will show that the variance grows by at most a factor of $w^3$.
That is, for any state $\psi$,
\begin{equation*}
    \LocalMeasurementVariance{U\psi U^\dagger} \le O\br{w^3}\cdot\LocalMeasurementVariance{\psi}.
\end{equation*}
Since the variance of a product state is upper bounded by one,
the variance of any quantum state $\psi$ produced by a $\twoDQAC$ circuit with width $w$ and depth $d$
is upper bounded by $\LocalMeasurementVariance{\psi} \le O\br{w^{3d}}$.
Immediately this leads to circuit lower bounds for constructing the cat state.
The local measurement variance growth of each $\twoDQAC$ layer is controlled in two steps:

\begin{itemize}
\item 
We follow the observables backward through the gates and bound the
change in their sum. A generalized controlled-$Z$ gate is a reflection
about a product state. When it meets several separate observables, the
changes associated with different regions act through orthogonal
components. This controls their combined norm without paying for the
gate's size. To turn this norm estimate into a variance bound, we
group the resulting observables into families with disjoint column
projections.

\item \textbf{Using geometry to bound the overlap.}
It would create too
much overlap if we simply enlarged every region by all gates. Thus, we instead make use of the geometrical structure by ordering gates according to their leftmost columns and
absorbing some into the regions before tracking the remaining changes.
Every remaining gate meeting a region contains its leftmost column,
so disjoint gate supports bound their number by the grid width.
The enlarged regions and correction terms also have overlap bounded
in terms of width. Together, these facts limit the variance growth
through a whole layer to a factor of $O(w^3)$. Iterating from a product input
and comparing with the target state's variance gives the width
lower bounds.
\end{itemize}




\paragraph{Upper bounds from Rosenthal's synthesis.}
For the unconditional upper bound, we use Rosenthal's depth-two
approximate nekomata synthesis~\cite{rosenthal:LIPIcs.ITCS.2021.32}.
Its targets fit in one column, with exponential workspace along the
rows. We make the synthesis exact using the amplification of Grier,
Morris, and Wu~\cite{grier2026mathsfqac0containsmathsftc0with}, implemented
with row and column operations, and convert it into a clean parity
circuit. Combining these small parity circuits recursively gives a
narrow upper bound. And the conversion between parity and fanout then gives a cat-state
synthesis. 

\subsection{Related Works}
\label{sec:related-works}


\paragraph{PARITY and $\mathsf{QAC}^0$.}
Moore~\cite{moore1999quantum} introduced $\QACzero$ as a quantum analog of $\mathsf{AC}^0$,and asked whether $\Parity\in\QACzero$. This question is central in quantum circuit complexity, because $\Parity$ is equivalent to several other functions under constant-depth
reductions, including $\mathsf{FANOUT}, \mathsf{MAJORITY},\mathsf{MOD}$~\cite{moor,hoyer2005quantum,6597759,Grier:2024xxt}. Thus, lower bounds for $\Parity$ were studied extensively under restrictions on the depth or on the number of
ancillas~\cite{10.5555/2011679.2011682,10.1016/j.ipl.2011.05.002,fenner2025tightboundsdepth2qaccircuits,NPVY24,ADOY24,joshi2025improvedlowerboundsqac0}.


Nadimpalli, Parham, Vasconcelos, and Yuen initiated the study of the Pauli spectrum
of single-output channels and proved low-degree concentration with
limited ancillas \cite{NPVY24}.  Their method combines lightcone
structure after deleting large gates with a bound on the deletion error
in a normalized Hilbert--Schmidt norm.  It yields exponentially small
average-case $\Parity$ advantage when the number of ancillas is at most
a sufficiently small multiple of $n^{1/d}$.  However, a large gate deletion induces an exponentially large error in terms of the ancillas.  This explains why unrestricted workspace defeats this
gate-erasure argument.

Anshu, Dong, Ou, and Yao moved beyond this regime by deriving
low-degree approximations in operator norm \cite{ADOY24}.  Their
depth-$d$ parity lower bound requires $n^{1+3^{-d}}$ ancillas.
Dong, Ou, and Yao sharpened the underlying approximate-degree exponent
from $1-3^{-d}$ to $1-2^{-d}$ using operator dilation, strengthening
the barely superlinear resource bounds
\cite{dong2025linearsizeqac0channelslearning}. \cite{ADOY24} also showed that improving their lower bound to
$n^{1+\exp(-o\left(d\right))}$ would exclude $\Parity$ from $\QACzero$.
Our width recursion gives a related route through geometric width:
it identifies a depth dependence that would suffice under polynomial
resource restrictions, while our unrestricted-ancilla synthesis
shows why such a restriction matters.
OpenAI recently released two manuscripts proving that $\Parity\notin\QACzero$~\cite{openai2026localization,openai2026trajectories}. Gao, Ji, and Liu later gave a simpler, self-contained proof of this result~\cite{gao2026simplifiedparity}. All these lower bounds only hold for polynomial-size $\QAC$ circuits. In contrast, our lower bound on the width  holds even for  $\twoDQAC$ circuits with unrestricted ancillary workspace, and we provide a construction showing that this lower bound is tight.

For unrestricted ancillas, the lower bounds were also studied.
Fenner, Grier, Pad\'e, and Thierauf excluded depth-two exact
parity computation even with unrestricted ancillas and garbage,
using polynomial irreducibility and the entangling behavior of
generalized controlled-$Z$ gates
\cite{fenner2025tightboundsdepth2qaccircuits}.  Rosenthal established
depth-two approximation lower bounds and related parity to synthesizing
nekomata states \cite{rosenthal:LIPIcs.ITCS.2021.32}.
Joshi, Tal, Vasconcelos, and Wright later excluded depth-three exact parity
with arbitrarily many ancillas, using quantum restrictions,
classical simulation, and monotonicity at small depths
\cite{joshi2025improvedlowerboundsqac0}.  They also strengthened the
depth-two approximation bounds through entropy and Fourier arguments.
Their depth restrictions and our width restrictions address different
ways of limiting the use of large gates; our bound continues to apply
as the number of layers increases.

\paragraph{$\mathsf{QAC}^0$ upper bounds}
The quantum advantages of $\QACzero$ were studied in several recent works.  Grier, Morris, and Wu proved
$\mathsf{TC}^0\subseteq\QACzero\circ\mathsf{NC}^0$, with the
classical preprocessing supplying many copies of the input, and obtained
decision separations from classical modular circuits
\cite{grier2026mathsfqac0containsmathsftc0with}.  Their exact
amplification is an ingredient of our upper bound. Vasconcelos gave a further separation from $\mathsf{AC}^0$
by computing the promise problem 2-fold Forrelation in $\QACzero$
on explicit inputs \cite{vasconcelos2026forrelation}.
Gheorghiu gave classically verifiable sampling tasks in
$\QACzero$ that are hard for classical polynomial-time algorithms
under the knowledge-of-lattice-point assumption and the strengthened adaptive-hardcore-bit
assumption \cite{gheorghiu2026verifiable}.
Foxman, Parham, Vasconcelos, and Yuen gave constant-depth synthesis of
approximate unitary designs and, under cryptographic assumptions, pseudorandom
unitaries \cite{foxman2025randomunitariesconstantquantum}.
Joshi and Vasconcelos gave exact synthesis of constant-weight Dicke
states in $\QACzero$, with arbitrary weight requiring additional
fanout \cite{vasconcelos2026constantdepthunitarypreparationdicke}.
Gretta, Gupta, and Joshi extended this to polylogarithmic weight in
$\QACzero$ and synthesized $D_{n,k}$ using fanout on
$\min\{k,n-k\}$ qubits \cite{gretta2026polylogdicke}.
Their later work also gave exact clean synthesis of every
$O(\log n)$-qubit state with polynomially many ancillas
\cite{gretta2026logarithmicstate}.

Constant-depth unitary synthesis further illustrates the power of large ancillary spaces. Strelchuk, Subramanian, and Weisz~\cite{strelchuk2026power} gave exact constant-depth constructions for computational-basis permutations, diagonal unitaries, and arbitrary pure-state preparation using $\QAC$ circuits, with unrestricted ancillas. Nehoran, Slote, and Yuen~\cite{nehoran2026all} showed that every $n$-qubit unitary can be implemented exactly in constant-depth $\QAC$ circuit with fanout gate and unrestricted ancillas.


\paragraph{Fourier concentration}
Gretta, Gupta, and Joshi~\cite{GGJ26} proved that PARITY $\notin
\mathsf{QAC}^0$ if and only if $\mathsf{QAC}^0$ circuits are Fourier-concentrated,
in analogy with the Linial--Mansour--Nisan theorem.  Their Fourier-to-state encoding
\cite{GGJ26} relates high-degree Fourier weight to the
probability of outcomes with large Hamming weight, while retaining constant probability on the all-zero outcome.  Our argument in
\cref{sec:fourier-concentration} follows a similar approach: phase
encoding gives a Fourier sampling distribution with
constant probability on the empty set.  We combine this property with the geometric output variance bound to control the Fourier tail.

\paragraph{Learning}
Related learning work ranges from quantum process tomography
\cite{Chuang_1997} to
learning shallow bounded-arity circuits \cite{10.1145/3618260.3649722},
and Pauli-based learning of channels with few ancillas \cite{NPVY24}.
Vasconcelos and Huang obtained quasipolynomial time and sample complexity
for learning $\QACzero$ unitaries with logarithmically many ancillas
\cite{vasconcelos2024learning}; Coble and Coudron studied the different
task of approximating output probabilities of geometrically local
circuits \cite{9719811}. We note that $\twoDQAC$ circuits have a low Schmidt rank and the Matrix product states connect bounded Schmidt ranks to classical simulation \cite{vidal2003slightly}. Using this MPS representation, we apply the learning bounds in \cite{khavari2021pseudodimension,hanneke2016optimal} to learn Boolean functions computed by $\twoDQAC$ circuits.


\subsection*{Acknowledgment} 
This work was supported by the National Natural Science Foundation of China (Grant Nos. 62332009 and 12347104), the Quantum Science and Technology—National Science and Technology Major Project (Grant No. 2021ZD0302901), the NSFC/RGC Joint Research Scheme (Grant No. 12461160276), the Natural Science Foundation of Jiangsu Province (No. BK20243060), the Fundamental and Interdisciplinary Disciplines Breakthrough Plan of the Ministry of Education of China (No. JYB2025XDXM118), the "111 Center"(No. B26023), and the Fundamental Research Funds for the Central Universities (Grant no. 2026300376).
Most of the work was done when Y.D. was a Ph.D. student at Nanjing University.
Y.D. was also funded by the European Union (ERC, DistQuant, 101266668). Views and opinions expressed are however those of the author(s) only and do not necessarily reflect those of the European Union or the European Research Council. Neither the European Union nor the granting authority can be held responsible for them.

\section{Preliminaries}
\label{sec:preliminaries}

\subsection{Boolean functions and notation}
\label{sec:boolean-preliminaries}

Write $[n]=\{1,\ldots,n\}$. Complements are taken in the ambient
index set or register, and $|x|$ denotes the Hamming weight of a bit
string $x$. Asymptotic notation hides absolute constants unless
dependence on fixed parameters is stated explicitly, and $\poly$ denotes
a polynomial bound. Unspecified logarithms may be taken in any fixed base.

For a Boolean function $f:\bits^n\to\bits$, write $h=(-1)^f$ for its
sign representation, $g=i^f$ for its phase encoding, and
$\mu=\E_x \Br{h(x)}$. Expectations over $x\in\bits^n$ are uniform unless
otherwise specified. Let $e_j$ be the string with a single one in
coordinate $j$. The sensitivity of $f$ at $x$ is
\[
 s_f(x)=\bigl|\{j\in[n]:f(x)\ne f(x\oplus e_j)\}\bigr|.
\]
Its average sensitivity, also called total influence, is
$\AverageSensitivity(f)=\E_x s_f(x)$.

For $S\subseteq[n]$, define
$\chi_S(x)=(-1)^{\sum_{j\in S}x_j}$. For a complex-valued function
$u$ on the cube, define $\widehat u(S)=\E_x[u(x)\chi_S(x)]$ and
$\norm u_2^2=\E_x|u(x)|^2$. The characters form an orthonormal basis,
giving the Fourier expansion and Parseval identity
\[
 \begin{aligned}
  u(x)&=\sum_{S\subseteq[n]}\widehat u(S)\chi_S(x),\\
  \norm u_2^2&=\sum_{S\subseteq[n]}|\widehat u(S)|^2.
 \end{aligned}
\]
In particular, $\widehat u(\varnothing)=\E u$, and real-valued
functions have real Fourier coefficients. The Fourier weight at
degrees at least $k$ is
$
 \FourierWeight{\ge k}{u}=\sum_{|S|\ge k}|\widehat u(S)|^2.
$
When $\norm u_2=1$, the probabilities $\Pr_u[S]=|\widehat u(S)|^2$
define a distribution on subsets of $[n]$, called Fourier sampling.
We use $\E_u$ and $\Var_u$ for expectation and variance under this
distribution. This applies to both $h$ and $g$; the sampled degree
is $|S|$.

\subsection{Quantum information}
\label{sec:quantum-preliminaries}

All Hilbert spaces in this work are finite-dimensional. A state $\rho$ is a
positive semidefinite operator of trace one; a unit vector $\ket\psi$
represents the pure state,
which we write as $\psi = \proj\psi$. We use $\ket x$ for a
computational-basis vector and $\ket S$ for the vector indexed by the
indicator string of $S\subseteq[n]$. The Pauli operators are denoted
$X,Z$, and the Hadamard gate by $H$. We write
$\ket+=(\ket0+\ket1)/\sqrt2$.

For a register $R$, the strings $0^R$ and $1^R$ are constant on $R$,
and the reduced state is $\rho_R=\Tr_{R^c}\rho$. An operator supported
on $R$ acts as the identity on its complement. Tensor-product symbols
are often suppressed between registers.

For vectors, $\norm v=\norm v_2$ is the Hilbert-space norm. For
operators, $\norm A$ is the operator norm and
$\norm A_1=\Tr\sqrt{A^\dagger A}$ is the trace norm. The Frobenius
norm $(\Tr A^\dagger A)^{1/2}$ bounds the operator norm. We use
trace-norm duality in the form
\[
 |\Tr(A^\dagger B)|\le\norm A\cdot\norm B_1.
\]
A Hermitian contraction is a Hermitian operator of
norm at most one.

The trace distance is
$\trdist{\rho}{\sigma}=\tfrac12\norm{\rho-\sigma}_1$.
For $\Delta=\rho-\sigma$, its variational forms are
\[
 \trdist{\rho}{\sigma}
 =\max_E|\Tr(E\Delta)|
 =\frac12\max_H|\Tr(H\Delta)|,
\]
where the maxima range over operators $0\le E\le I$ and Hermitian contractions
$H$, respectively. Trace distance is invariant under unitaries,
cannot increase under partial trace, and bounds the change in any
measurement probability.

We use squared fidelity with a pure target,
$\fidelity{\rho}{\proj\psi}=\bra\psi\rho\ket\psi$. The estimates
\[
 \begin{aligned}
  \trdist{\rho}{\proj\psi}
  &\le\sqrt{1-\fidelity{\rho}{\proj\psi}},\\
  \trdist{\proj\psi}{\proj\phi}
  &=\sqrt{1-|\braket{\psi}{\phi}|^2}
   \le\norm{\ket\psi-\ket\phi}
 \end{aligned}
\]
relate this fidelity and vector approximation to trace distance.
The Schmidt rank of a pure state across a bipartition is the least
number of product vectors whose sum equals the state.

\subsection{Circuit models}
\label{sec:qac-circuits}

A generalized controlled-$Z$ gate on a set $S$ of qubits is
$\CZ_S=I-2\proj{1^S}$, with the identity understood on the remaining
qubits. Conjugating one qubit by Hadamards gives a generalized Toffoli
gate, denoted $\mathsf{CNOT}$,
which maps $\ket{x,b}$ to $\ket{x,b\oplus\mathbf1\{x=1^n\}}$
for $x\in\bits^n$ and $b\in\bits$.

\begin{definition}[$\QAC$ circuits]
\label{def:circuit-model}
A depth-$d$ $\QAC$ circuit has the form
\[
 C=L_dM_d\cdots L_1M_1L_0,
\]
where each $L_j$ is a tensor product of one-qubit unitaries and each
$M_j$ is a product of generalized controlled-$Z$ gates with pairwise
disjoint supports. There is no restriction on the arity of these gates.
\end{definition}

\noindent \textbf{Convention (nonuniform circuits)}. 
Throughout this work, circuit families are nonuniform. Arbitrary one-qubit unitaries may be chosen separately for each input size, with arbitrary real parameters. We impose no finite gate-set or uniformity requirement unless explicitly stated.

Throughout this work, circuit families are nonuniform. We impose no finite gate-set or uniformity requirement unless explicitly stated.

Unless otherwise stated, ancillas are initialized to be $\ket{0^a}$. “Unrestricted ancillas” means that no asymptotic upper bound is imposed on $a$, although every individual circuit is finite.

A product reflection on $G$ has the form $R=I-2\Pi$, where
$\Pi=\proj\phi\otimes I_{G^c}$ and
$\ket\phi=\bigotimes_{v\in G}\ket{\phi_v}$ is a product unit vector.

Conjugating a generalized controlled-$Z$ gate by one-qubit unitaries
gives a product reflection with the same support.
For the circuit in \cref{def:circuit-model}, set $V_j=L_j\cdots L_0$
for $0\le j\le d$ and $U_j=V_{j-1}^\dagger M_jV_{j-1}$ for
$1\le j\le d$. Then
\[
 C=V_dU_d\cdots U_1.
\]
Each $U_j$ is a layer of product reflections with disjoint supports,
and $V_d$ is a tensor product of one-qubit unitaries. This
representation preserves depth, supports, and qubit count.



\begin{definition}[$\QACzero$]
\label{def:qac-zero}
A $\QACzero$ family consists of constant-depth $\QAC$ circuits
$\{C_n\}_{n\ge1}$ with polynomially many qubits and gates in the input
size $n$. As a decision class, $\QACzero$ denotes the languages
computed by such families with bounded worst-case error.
\end{definition}

We distinguish circuit-level statements from complexity-class statements. A depth-$d$~$\QAC$ circuit may have arbitrarily many ancillas initialized to product states. By contrast, a $\QACzero$ family is required to use polynomially many qubits and gates. Thus, a constant-depth $\QAC$ construction with superpolynomial workspace is not a $\QACzero$ family.

\begin{definition}[Geometrically restricted $\QAC$ circuits]
\label{def:geometric-qac}
A $\twoDQAC$ circuit is a $\QAC$ circuit whose qubits fill a rectangle
$[w]\times[h]$ of $w$ rows and $h$ columns. Each multi-qubit gate has
contiguous support in one row or column. We call $w$ its width and $h$ its length, and write
$[w,h]$-$\twoDQAC$ for this layout. Width one gives the $\oneDQAC$ model.
The classes $\twoDQACzero$ and $\oneDQACzero$ consist of
constant-depth, polynomial-size families in these geometries, with
the same decision convention as $\QACzero$.
\end{definition}
It is worth noticing that an allowed interval may have an arbitrarily large diameter. Thus, “geometrically local” here refers to connected interval support, rather than bounded interaction range.

We distinguish Boolean function computation, clean unitary synthesis,
and state synthesis by their output guarantees.

\begin{definition}[Computing a Boolean function]
\label{def:computation-models}
For $f:\bits^n\to\bits$, apply $C$ to $\ket x$ and its prescribed
ancillas, and let $g_C(x)$ be the outcome of measuring the designated
output qubit in the computational basis. The circuit computes $f$
with success probability at least $p$ if
$\Pr[g_C(x)=f(x)]\ge p$ for every $x\in\bits^n$.

Exact computation has $p=1$. Unless stated otherwise, approximate
(bounded-error) computation means success probability at least $2/3$
on every input. The remaining registers may contain input-dependent
garbage.
\end{definition}

Average-case success averages
over uniform $x$ as well as the measurement outcome. The coherent
version of $f$ is the unitary
\[
 U_f\ket{x,b}=\ket{x,b\oplus f(x)},
\]
where $b$ is a one-qubit accumulator.

\begin{definition}[Unitary synthesis]
\label{def:clean-implementation}
An exact clean implementation of a unitary $U$ satisfies
$C\ket{\phi,0^a}=(U\ket\phi)\ket{0^a}$ for every logical input
$\ket\phi$, up to a fixed global phase. A $p$-approximate clean
implementation satisfies
\[
 \left|\bigl(\bra\phi U^\dagger\otimes\bra{0^a}\bigr)
 C\ket{\phi,0^a}\right|^2\ge p
\]
for every $\ket\phi$. This is a squared-fidelity guarantee on the
entire output, including the ancillary workspace.
\end{definition}

A $p$-approximate clean implementation of $U_f$ computes $f$ with
success probability at least $p$ when the accumulator starts at zero.
Decision success alone does not require coherent phases or restored
ancillas.

\begin{definition}[State synthesis]
\label{def:state-synthesis}
Let $C$ act on $n+a$ qubits initialized to $\ket{0^{n+a}}$, and let
$\rho_T$ be its reduced output on an $n$-qubit target register $T$.
It synthesizes $\ket\psi$ with squared fidelity at least $p$ if
$\bra\psi\rho_T\ket\psi\ge p$. Clean synthesis requires instead
\[
 \left|\bra{\psi,0^a}C\ket{0^{n+a}}\right|^2\ge p.
\]
Taking $p=1$ gives exact synthesis in either case.
\end{definition}

The first guarantee is called target-marginal fidelity and leaves
the remaining registers unrestricted.

\subsection{Parity, fanout, and cat states}
\label{sec:parity-cat-definitions}

The parity function is $\Parity_n(x)=\bigoplus_{j=1}^n x_j$.
Its associated unitary $U_{\oplus_n}$ and the fanout unitary $U_{F_n}$
act as
\[
 \begin{aligned}
  U_{\oplus_n}\ket{x,b}&=\ket{x,b\oplus\Parity_n(x)},\\
  U_{F_n}\ket{b,x_1,\ldots,x_n}
  &=\ket{b,x_1\oplus b,\ldots,x_n\oplus b}.
 \end{aligned}
\]
Hadamard conjugation interchanges parity and fanout. Applying fanout
to $\ket+\ket{0^n}$ synthesizes an $(n+1)$-qubit cat state, where
\[
 \Cat n=\frac{\ket{0^n}+\ket{1^n}}{\sqrt2}.
\]

An $n$-nekomata allows an auxiliary state in each branch:
\[
 \ket\nu=\frac{\ket{0^n}\ket{\eta_0}
                  +\ket{1^n}\ket{\eta_1}}{\sqrt2},
\]
where $\ket{\eta_0}$ and $\ket{\eta_1}$ are normalized. Its target
register is a cat state when the two auxiliary states coincide.
A $p$-approximate nekomata has squared fidelity at least $p$ with
some state of this form, including the auxiliary register.
The geometric conversions and exact amplification are proved in
\cref{sec:cat-fanout-exactification}.

Clean parity synthesis gives cat-state synthesis through the following
conversions.
\begin{restatable}{proposition}{parityfanoutreductions}
\label{prop:parity-fanout-reductions}
Let $d\ge1$. The following conversions hold for depth-$d$
$\left[w,h\right]$-$\twoDQAC$ circuits:
\begin{itemize}
\item A clean implementation of the parity unitary on $n$ controls
can be converted into a clean implementation of the fanout unitary,
and conversely, with the same depth and qubit layout. Either circuit
also synthesizes an $(n+1)$-qubit cat state on its logical qubits.
\item If a circuit computes a Boolean function $f$ exactly and its
output qubit lies on the boundary of the rectangle, it gives a clean
implementation of $U_f$ by a depth-$O\left(d\right)$
$\left[w+O\left(1\right),h+O\left(1\right)\right]$-$\twoDQAC$ circuit.
\item If a circuit synthesizes an exact $n$-nekomata whose targets occupy
consecutive positions in a boundary column, it gives a clean parity
implementation on $n$ controls by a depth-$O\left(d\right)$
$\left[w+O\left(1\right),h+O\left(1\right)\right]$-$\twoDQAC$ circuit. The new input qubits lie in an
adjacent column.
\end{itemize}
\end{restatable}
The parity--fanout conversion and the resulting cat synthesis also
preserve approximate clean squared fidelity.

We also use the following amplification result, for synthesis depth
$d\ge1$.
\begin{restatable}{lemma}{endpointamplification}
\label{lem:endpoint-amplification}
Let $C$ be a depth-$d$ $\QAC$ circuit synthesizing a $p$-approximate
$n$-nekomata, where $p>1/2$. Then, an exact $(n+2)$-nekomata can be
synthesized in depth $O\left(d/(\sqrt{2p}-1)\right)$.
Moreover, if $C$ is a $\left[w,h\right]$-$\twoDQAC$ circuit whose targets fill a
boundary column, exact synthesis is possible with a
$\left[w+O\left(1\right),h+O\left(1\right)\right]$-$\twoDQAC$ circuit with the same depth bound.
For a fixed constant $p>1/2$, this conversion preserves
polynomial size and constant depth, and hence also $\twoDQACzero$.
\end{restatable}

The proofs of \cref{prop:parity-fanout-reductions,lem:endpoint-amplification}
are given in Appendix~\ref{sec:cat-fanout-exactification}.
\par

\section{Lower bounds on \texorpdfstring{$\twoDQAC$}{2D-QAC}}
\label{sec:geometric-restrictions}
\label{sec:lower-bounds}

In this section, we first show that two-dimensional geometry preserves the power of
unrestricted $\QAC$ circuits. Then we investigate narrow $\twoDQAC$ circuits, i.e., the circuits with small width.  The local measurement variance bound limits cat-state
synthesis independently of ancillary resources and also applies to
Dicke states. We then relate Boolean phase-state variance to sensitivity
and derive a Fourier tail estimate. These estimates give lower bounds
for exact Boolean computation in narrow circuits. 
We also give a width recursion relating depth and width under polynomial resource bounds. In one dimension, we obtain linear depth lower
bounds for PARITY, MAJORITY, and fixed-modulus MOD with arbitrary input
placements


\subsection{Exactly Simulating General \texorpdfstring{$\QAC$}{QAC} Circuits with \texorpdfstring{$\twoDQAC$}{2D-QAC} Circuits}
\label{sec:two-dimensional-simulation}

We show how to simulate a general $\QAC$ circuit with all-to-all
connectivity by a $\twoDQAC$ circuit. The idea is to move the qubits
involved in each gate to a common row, apply the gate there, and move
them back. Gates in the same layer have disjoint supports, so these
operations can be performed in parallel using a separate row for each
gate. Since a layer on $n$ qubits contains at most $n$ multi-qubit
gates, $n$ additional rows suffice for the simulation.

\begin{theorem}
\label{thm:two-dimensional-simulation}
Let $U$ be a unitary implemented by a depth-$d$ $\QAC$ circuit on
$n$ qubits. There is a depth-$7d$ $\left[n+1,n\right]$-$\twoDQAC$ circuit $V$
that exactly simulates $U$: for every $n$-qubit
input state $\ket\phi$,
\[
 V\bigl(\ket\phi\otimes\ket{1^{n^2}}\bigr)
 =\bigl(U\ket\phi\bigr)\otimes\ket{1^{n^2}}.
\]

Consequently, $\twoDQACzero=\QACzero$.
\end{theorem}

\begin{proof}
Place the logical qubits in row $0$, one per column, and initialize $n$
additional rows to $\ket1$.  Consider a multi-qubit layer with disjoint
supports $S_1,\ldots,S_t$, where $t\le n$.  Assign gate $i$ to row $i$.
For each $j\in S_i$, swap the qubits at $(0,j)$ and $(i,j)$.

These swaps have constant-depth geometric implementations.  A CNOT
between the endpoints of an interval can be replaced by a Toffoli whose
other controls are all the intermediate qubits, provided those qubits are
in $\ket1$.  This replacement acts as the desired CNOT on arbitrary
states of its two endpoints and leaves the intermediate qubits unchanged.
Decompose an endpoint swap into three such CNOTs.  There
is at most one swap in each column, since the $S_i$ are disjoint.
Consequently, all swaps run in parallel, and all intermediate controls
remain $\ket1$ throughout this operation.

Row $i$ now contains the logical qubits in $S_i$ needed by gate $i$.  All other qubits in that row are $\ket1$.
Applying controlled-$Z$ to the interval from the leftmost to the
rightmost qubit of $S_i$ therefore implements $\CZ_{S_i}$.  These gates
occupy different rows.  Reversing the swaps returns every logical qubit
to row $0$ and restores the workspace to $\ket1$.  The next layer can
reuse it.  Apply each original one-qubit layer on the data row.

\Cref{fig:two-dimensional-simulation} illustrates this simulation of
one layer.

\begin{figure}[htp]
\centering
\begin{tikzpicture}[
    basic box/.style={
        draw, thick, rounded corners=3pt,
        align=center, font=\large, inner sep=0pt
    },
    s1 box/.style={basic box, fill=blue!20},
    s2 box/.style={basic box, fill=green!20},
    grid line/.style={draw=black!22, line width=0.35pt},
    grid border/.style={draw=black!45, line width=0.6pt},
    myarrow/.style={->, >=stealth, thick}
]

\newcommand{\drawswap}[3]{
    \draw[black] (#1, #2) -- (#1, #3);
    \node[scale=1.5] at (#1, #2) {$\times$};
    \node[scale=1.5] at (#1, #3) {$\times$};
}

\begin{scope}[xshift=0cm, scale = 0.7, transform shape]

    \draw[grid border] (0,0) rectangle (6,3);
    \foreach \x in {1,2,3,4,5} \draw[grid line] (\x, 0) -- (\x, 3);
    \foreach \y in {1,2} \draw[grid line] (0, \y) -- (6, \y);

    \drawswap{0.5}{0.5}{1.5}
    \drawswap{1.5}{0.5}{1.5}
    \drawswap{2.5}{0.5}{2.5}
    \drawswap{3.5}{0.5}{2.5}
    \drawswap{4.5}{0.5}{2.5}
    \drawswap{5.5}{0.5}{1.5}

\end{scope}

\begin{scope}[xshift=5cm, scale = 0.7, transform shape] 

    
    \node[s1 box, minimum width=1.8cm, minimum height=0.9cm] (L_S1a) at (1, 5.5) {$S_1$};
    \draw[thick] (0.5, 5.05) -- (0.5, 4.8); 
    \draw[thick] (1.5, 5.05) -- (1.5, 4.8); 
    \draw[thick] (0.5, 5.95) -- (0.5, 6.2); 
    \draw[thick] (1.5, 5.95) -- (1.5, 6.2); 
    
    \node[s2 box, minimum width=2.8cm, minimum height=0.9cm] (L_S2) at (3.5, 5.5) {$S_2$};
    \foreach \x in {2.5, 3.5, 4.5} {
        \draw[thick] (\x, 5.05) -- (\x, 4.8); 
        \draw[thick] (\x, 5.95) -- (\x, 6.2); 
    }

    \node[s1 box, minimum width=0.9cm, minimum height=0.9cm] (L_S1b) at (5.5, 5.5) {$S_1$};
    \draw[thick] (5.5, 5.05) -- (5.5, 4.8); 
    \draw[thick] (5.5, 5.95) -- (5.5, 6.2); 

    \node[thick, scale=2.5] at (3, 3.9) {$\Downarrow$};

    \fill[green!20] (2,2) rectangle (5,3);
    \fill[blue!20] (0,1) rectangle (6,2);

    \draw[grid border] (0,0) rectangle (6,3);
    \foreach \x in {1,2,3,4,5} \draw[grid line] (\x, 0) -- (\x, 3);
    \foreach \y in {1,2} \draw[grid line] (0, \y) -- (6, \y);

    \node at (3.5, 2.5) {\large $S_2$};
    \node at (3.0, 1.5) {\large $S_1$};

\end{scope}

\begin{scope}[xshift=10cm, scale = 0.7, transform shape]

    \draw[grid border] (0,0) rectangle (6,3);
    \foreach \x in {1,2,3,4,5} \draw[grid line] (\x, 0) -- (\x, 3);
    \foreach \y in {1,2} \draw[grid line] (0, \y) -- (6, \y);

    \drawswap{0.5}{0.5}{1.5}
    \drawswap{1.5}{0.5}{1.5}
    \drawswap{2.5}{0.5}{2.5}
    \drawswap{3.5}{0.5}{2.5}
    \drawswap{4.5}{0.5}{2.5}
    \drawswap{5.5}{0.5}{1.5}

\end{scope}

\end{tikzpicture}
\caption{Simulation of one multi-qubit layer. The left and right panels show the column swaps; the middle panel shows the gates on their assigned workspace rows. Logical qubits start and finish on the bottom data row.}
\label{fig:two-dimensional-simulation}
\end{figure}

The forward swaps take three layers, the row gates take one, and the
reverse swaps take another three. Repeating this simulation for all
$d$ layers gives depth $7d$ on a $(n+1)$-by-$n$ rectangle.
Local $X$ gates convert the workspace
initialization and restoration to the all-zero convention if desired.
The simulation introduces only $O\left(dn\right)$ multi-qubit gates, so it
preserves polynomial resources for a constant-depth family.  The reverse
class inclusion follows by forgetting the geometric restriction.
\end{proof}

\subsection{Local Measurement Variance}
\label{sec:local-measurement-variance}


We introduce local measurement variance in
\cref{def:local-measurement-variance} to quantify correlations between
measurements on separated regions. We bound this quantity for product
and cat states in \cref{lem:product-variance,lem:cat-variance}, and for
Dicke states in \cref{prop:dicke-variance}. These bounds are used to
prove the width lower bounds in \cref{thm:width-lower-bound}.

Let $\Lambda=[w]\times[h]$ denote the grid of qubit positions,
with rows indexed by the first coordinate
and columns by the second.  A connected region is a nonempty subset
$R\subseteq\Lambda$ whose induced nearest-neighbor graph is connected,
where adjacent qubit positions differ by one in exactly one coordinate.
For any region $R\subseteq\Lambda$, define its \emph{column projection} by
\[
 I(R)=\{c\in[h]:\text{there exists }r\in[w]\text{ with }(r,c)\in R\}.
\]
This is an integer interval whenever $R$ is connected.
We call a family of nonempty connected regions \emph{separate}
if their column projections are pairwise disjoint.  Note that this
separation condition implies that the regions themselves are pairwise
disjoint.  For a state $\rho$ and a Hermitian observable $H$, define
\[
 \Var_\rho(H)=\Tr(\rho H^2)-\bigl(\Tr(\rho H)\bigr)^2.
\]

\begin{definition}[Local measurement variance]
\label{def:local-measurement-variance}
The local measurement variance of a state $\rho$ on $\Lambda$ is
\[
 \LocalMeasurementVariance{\rho}
 =\max_{\mathcal R,\,(H_i),\,(c_i)}
      \Var_\rho\left(\sum_{i=1}^m c_iH_i\right).
\]
The maximum ranges over families $\mathcal R=\{R_1,\ldots,R_m\}$
of nonempty regions $R_i\subseteq\Lambda=[w]\times[h]$, each connected
in the two-dimensional nearest-neighbor graph, with \textit{pairwise disjoint}
column projections $I(R_i)$; Hermitian contractions $H_i$ supported
on $R_i$; and real coefficients satisfying $\sum_i c_i^2=1$.
\end{definition}

We use connected regions because tracing a region backward through a
circuit layer enlarges it by horizontal or vertical gate intervals
that meet it. The enlarged region remains connected, allowing us to
track the observables through successive layers.

\begin{remark}
\label{rem:marginal-variance}
For variance calculations, we may always restrict our attention to pure
states by adjoining a purifying register.  This is because for a pure state $\ket\psi$ on $T\otimes R$ that
purifies $\rho_T$,
every Hermitian observable $H$ on $T$ satisfies
\[
 \begin{aligned}
  \Var_{\rho_T}(H)
  &=\bra\psi(H^2\otimes I_R)\ket\psi
    -\bigl(\bra\psi(H\otimes I_R)\ket\psi\bigr)^2\\
  &=\Var_\psi(H\otimes I_R).
 \end{aligned}
\]
In this work we may consider mixed quantum states, which are produced by discarding the ancillas of the output of a quantum circuit.
In this case, variance bounds on the full output also apply to measurements on the target marginal.
\end{remark}

We next bound the local measurement variance of product states and cat
states, beginning with product states.

\begin{lemma}
\label{lem:product-variance}
Every product state $\rho=\bigotimes_{v\in\Lambda}\rho_v$ satisfies
$\LocalMeasurementVariance{\rho}\le1$.
\end{lemma}

\begin{proof}
Fix a separate family, observables $H_i$, and coefficients $c_i$.
Expanding the variance, we obtain
\[
 \begin{aligned}
  \Var_\rho\left(\sum_i c_iH_i\right)
  &=\sum_{i,j}c_ic_j\cdot\Tr(\rho H_iH_j)
    -\left(\sum_i c_i\cdot\Tr(\rho H_i)\right)^2\\
  &=\sum_i c_i^2\cdot\left(\Tr(\rho H_i^2)
    -\bigl(\Tr(\rho H_i)\bigr)^2\right)\\
  &\quad+\sum_{i\ne j}c_ic_j\cdot\left(\Tr(\rho H_iH_j)
    -\Tr(\rho H_i)\cdot\Tr(\rho H_j)\right)\\
  &=\sum_i c_i^2\cdot\Var_\rho(H_i)\\
  &\le\sum_i c_i^2=1,
 \end{aligned}
\]
where the third equality uses the fact that the product state factors
across the disjoint supports, giving, for $i\ne j$,
\[
 \Tr(\rho H_iH_j)
 =\Tr(\rho H_i)\cdot\Tr(\rho H_j).
\]
The inequality uses $\Var_\rho(H_i)\le\Tr(\rho H_i^2)\le1$,
which follows from $\norm{H_i}\le1$.  Taking the maximum proves the claim.
\end{proof}

We now compute the variance for a cat state by measuring one qubit
in each occupied column.

\begin{lemma}
\label{lem:cat-variance}
The $n$-qubit cat state on a width-$w$ grid satisfies
\[
 \LocalMeasurementVariance{\mathrm{cat}_n}\ge\frac nw.
\]
\end{lemma}

\begin{proof}
Choose one qubit $j_i$ from each of the $m$ occupied columns and set
$H=\frac1{\sqrt m}\sum_{i=1}^m Z_{j_i}$.
The selected singletons form a separate family, and the coefficients
in $H$ have squared sum one.  Measuring $H$ on $\Cat n$ gives
$\sqrt m$ and $-\sqrt m$ with equal probability, so
\[
 \LocalMeasurementVariance{\mathrm{cat}_n}
 \ge\Var_{\mathrm{cat}_n}(H)=m\ge\frac nw,
\]
where the last inequality uses the fact that each column contains
at most $w$ of the $n$ qubits.
\end{proof}

For integers $0\le k\le n$, the Dicke state is the uniform
superposition of all strings with Hamming weight $k$:
\[
 \ket{D_{n,k}}=\binom nk^{-1/2}\sum_{|x|=k}\ket x.
\]
Its total Pauli-$Z$ observable has the fixed value $n-2k$, so its
variance in that direction vanishes. We instead use Pauli-$X$
correlations: flipping two unequal bits preserves the Hamming weight.

\begin{proposition}
\label{prop:dicke-variance}
For every $n\ge1$ and $0\le k\le n$, the Dicke state on a width-$w$
grid satisfies
\[
 \LocalMeasurementVariance{D_{n,k}}
 \ge\frac{n+2k(n-k)}{nw}.
\]
\end{proposition}

\begin{proof}
Since $X_j$ changes Hamming weight,
$\bra{D_{n,k}}X_j\ket{D_{n,k}}=0$.
For $n\ge2$ and $i\ne j$, $X_iX_j$ preserves the weight of a basis
string precisely when $x_i\ne x_j$. This holds for a fraction
$2k(n-k)/(n(n-1))$ of the weight-$k$ strings, so
$\bra{D_{n,k}}X_iX_j\ket{D_{n,k}}=2k(n-k)/(n(n-1))$.
Together with $X_j^2=I$, this gives
\[
 \Var_{D_{n,k}}\left(\sum_{j=1}^nX_j\right)=n+2k(n-k).
\]
For $n=1$, both Dicke states are basis states, and the same identity
holds.

To relate this variance to local measurement variance, let $\rho$ be
any state with target register $T$ and marginal $\rho_T$. Consider a sum
$A=\sum_{j=1}^n A_j$ of Hermitian single-qubit contractions on the
targets. For each column $c$, let $T_c$ be its target qubits, and set
$t_c=|T_c|$ and $A_c=\sum_{j\in T_c}A_j$. Then $A=\sum_c A_c$ and
$\norm{A_c}\le t_c$. The nonempty column operators $A_c/t_c$ are
contractions supported on separate connected columns. The definition
of local measurement variance therefore gives
\[
 \Var_{\rho_T}(A)
 \le\LocalMeasurementVariance{\rho}\cdot\sum_ct_c^2
 \le nw\cdot\LocalMeasurementVariance{\rho},
\]
where the last inequality uses $t_c\le w$ and $\sum_ct_c=n$.
Taking $A_j=X_j$ proves the claim, including when the Dicke state is
the target marginal of a larger state.
\end{proof}

\subsection{Lower bound for \texorpdfstring{$\twoDQAC$}{2D-QAC} states synthesis}
\label{sec:width-lower-bound}

In this subsection, \cref{thm:output-variance} bounds the local
measurement variance of states synthesized by narrow $\twoDQAC$ circuits,
and \cref{thm:width-lower-bound} gives width lower bounds for approximate
cat-state synthesis and exact Dicke-state synthesis.  The key estimate is
\cref{lem:variance-layer-growth}, which bounds the increase in local
measurement variance through a single circuit layer.

We work with layers of product reflections, using the circuit
representation described in \cref{sec:qac-circuits}.
The final tensor product of one-qubit unitaries preserves the local
measurement variance, since conjugation maps the observables in its
definition bijectively to themselves.

\begin{theorem}
\label{thm:output-variance}
There is an absolute constant $K\ge2$ such that every state
$\ket\psi=C\ket{0^n}$ synthesized by a depth-$d$, width-$w$ $\twoDQAC$
circuit on $n$ qubits satisfies
\[
 \LocalMeasurementVariance{\psi}\le(Kw^3)^d.
\]
\end{theorem}

\begin{theorem}
\label{thm:width-lower-bound}
Let $C$ be a depth-$d\ge1$, width-$w$ $\twoDQAC$ circuit starting
from $\ket{0^{n+a}}$ on $n$ target qubits and $a$ ancillas.
\begin{itemize}
\item If $C$ synthesizes $\Cat n$ with target-marginal squared fidelity
at least $1-\eta^2$, where $0\le\eta\le1/8$, then
\[
 w=\Omega\left(n^{1/(3d+1)}\right).
\]
\item For $0\le k\le n$, if $C$ synthesizes $\ket{D_{n,k}}$ exactly, then
\[
 w=\Omega\left((1+\min\{k,n-k\})^{1/(3d+1)}\right).
\]
\end{itemize}
\end{theorem}

The following lemma bounds the increase through a single layer.  Let $\mathcal G$
be its set of gate labels.  For $g\in\mathcal G$, write
$R_g=I-2\Pi_g$ for the product reflection and $G_g$ for its support.

\begin{restatable}{lemma}{lemlayergrow}
\label{lem:variance-layer-growth}
Let $U=\prod_{g\in\mathcal G}R_g$ be a layer of product reflections
with pairwise disjoint horizontal or vertical interval supports $G_g$ in a
width-$w$ grid.  Every pure state $\ket\psi$ satisfies
\[
 \LocalMeasurementVariance{U\psi U^{\dagger}}
 \le O\left(w^3\right)\cdot\LocalMeasurementVariance{\psi}.
\]
\end{restatable}

We first use \cref{lem:variance-layer-growth} to prove \cref{thm:output-variance}.

\begin{proof}[Proof of \cref{thm:output-variance}]
Using the representation in \cref{sec:qac-circuits}, write
$
 C=V_dU_d\cdots U_1
$,
where each $U_j$ is a layer of product reflections with pairwise
disjoint horizontal or vertical interval supports, and $V_d$ is a
tensor product of one-qubit unitaries.  Set $\ket{\psi_0}=\ket{0^n}$
and $\ket{\psi_j}=U_j\ket{\psi_{j-1}}$ for $1\le j\le d$.
Choose $K\ge2$ so that the factor in \cref{lem:variance-layer-growth}
is at most $Kw^3$.  Applying this lemma to each layer gives
\[
 \LocalMeasurementVariance{\psi_j}
 \le Kw^3\cdot\LocalMeasurementVariance{\psi_{j-1}}.
\]
Since $\ket\psi=V_d\ket{\psi_d}$, invariance under one-qubit
unitaries and \cref{lem:product-variance} yield
\[
 \LocalMeasurementVariance{\psi}
 =\LocalMeasurementVariance{\psi_d}
 \le(Kw^3)^d\cdot\LocalMeasurementVariance{0^n}
 \le(Kw^3)^d.
\]
\end{proof}

The following corollary follows from \cref{lem:variance-layer-growth}
by iterating over the circuit layers and grouping the target observables
by column.

\begin{corollary}
\label{cor:additive-variance-bound}
Let $\ket\psi=C\ket{0^{n+a}}$ be the output of a depth-$d$, width-$w$
$\twoDQAC$ circuit. For any sum $A=\sum_{j=1}^n A_j$ of Hermitian
single-qubit contractions on distinct target qubits,
\[
 \Var_\psi(A)\le K^d\cdot n\cdot w^{3d+1},
\]
where $K$ is the constant from \cref{thm:output-variance}.
\end{corollary}

\begin{proof}
Let $t_c$ count the target qubits in column $c$.
The column normalization from the proof of \cref{prop:dicke-variance}
and \cref{thm:output-variance} give
\[
 \begin{aligned}
  \Var_\psi(A)
  &\le\LocalMeasurementVariance{\psi}\cdot\sum_ct_c^2\\
  &\le(Kw^3)^d\cdot\sum_ct_c^2
   \le K^d\cdot n\cdot w^{3d+1},
 \end{aligned}
\]
where the last inequality uses $t_c\le w$ and $\sum_ct_c=n$.
\end{proof}

We next prove \cref{thm:width-lower-bound} by comparing the output
variance bound with the variance of the target states.

\begin{proof}[Proof of \cref{thm:width-lower-bound}]
Write $\ket\psi=C\ket{0^{n+a}}$ for the full output state.
Let $\rho_T=\Tr_{T^c}\proj\psi$ be the target marginal and take $H$
from the proof of \cref{lem:cat-variance}, with $m\ge n/w$ occupied
target columns.

For cat state synthesis, we give a lower bound of the local measurement variance to the final state. For the first moment,
\[
 \begin{aligned}
  \left|\Tr(\rho_T H)\right|
  &=\left|\Tr\bigl((\rho_T-\proj{\mathrm{cat}_n})H\bigr)\right|\\
  &\le\norm{\rho_T-\proj{\mathrm{cat}_n}}_1 \cdot \norm H\\
  &\le2 \eta \sqrt m,
 \end{aligned}
\]
where the equality uses $\bra{\mathrm{cat}_n}H\Cat n=0$.
For the last inequality, the fidelity hypothesis gives
$\norm{\rho_T-\proj{\mathrm{cat}_n}}_1\le2 \eta$,
and $\norm H=\sqrt m$.
For the second moment,
\[
 \begin{aligned}
  \Tr(\rho_T H^2)
  &\ge m \bra{\mathrm{cat}_n}\rho_T\Cat n\\
  &\ge m (1-\eta^2),
 \end{aligned}
\]
where the first inequality uses that
$H^2 \,-\,m \proj{\mathrm{cat}_n}$ is positive semidefinite because
$H^2\Cat n=m \Cat n$.
The choice of $H$ and \cref{rem:marginal-variance} give
\[
 \begin{aligned}
  \LocalMeasurementVariance{\psi}
  &\ge\Var_{\rho_T}(H)\\
  &\ge m (1-\eta^2)-(2 \eta \sqrt m)^2
   =m (1-5 \eta^2)\\
  &\ge\frac nw (1-5 \eta^2)
   \ge\frac{59n}{64w},
 \end{aligned}
\]
where we used $m\ge n/w$ and $\eta\le1/8$.
Combining this with \cref{thm:output-variance} gives
\[
 \frac{59n}{64}
 \le K^d w^{3d+1}.
\]

For exact Dicke-state synthesis, \cref{prop:dicke-variance} and
\cref{thm:output-variance}, applied to the full output state, give
\[
 \frac{n+2 k (n-k)}{n w}
 \le\LocalMeasurementVariance{\psi}\le(K w^3)^d.
\]
Multiplying by $w$ gives
\[
 1+\frac{2 k (n-k)}n\le K^d w^{3d+1}.
\]
Simplifying gives the claimed bound.
\end{proof}

We begin the proof of \cref{lem:variance-layer-growth} with two
auxiliary lemmas.

For a fixed choice of Hermitian operators $H_i$ and connected supporting
regions $D_i$, the next inequality bounds $\Var_\psi(\sum_i H_i)$
in terms of the local measurement variance.

\begin{lemma}
\label{lem:overlap-variance}
Let $H_1,\ldots,H_m$ be nonzero Hermitian operators supported on nonempty
connected regions $D_1,\ldots,D_m\subseteq\Lambda=[w]\times[h]$.
If each column index belongs to at most $r$ of the intervals $I(D_i)$,
then every pure state $\ket\psi$ satisfies
\[
 \Var_\psi\left(\sum_i H_i\right)
 \le r\cdot\LocalMeasurementVariance{\psi}\cdot\sum_i\norm{H_i}^2.
\]
\end{lemma}

\begin{proof}
The idea is to partition the intervals $I(D_i)$ into at most $r$ groups
of pairwise disjoint intervals.  Within each group, the regions have
disjoint column projections, so they are disjoint and form a separate
family.  We bound each group's variance by the local measurement
variance and then combine these bounds.

Process the intervals in nondecreasing order of their left endpoints,
greedily placing each in any group whose intervals are disjoint from it.
If no group among $r$ were available, each would contain a previous
interval covering the current left endpoint.  Together with the current
interval, these would give $r+1$ intervals containing that column index,
contradicting the overlap bound.

Write $J_a$ for the indices of the Hermitian operators in group $a$, and set
$c_a=\left(\sum_{i\in J_a}\norm{H_i}^2\right)^{1/2}$.
Expanding the variance within group $J_a$, we obtain
\[
 \begin{aligned}
  \Var_\psi\left(\sum_{i\in J_a}H_i\right)
  &=c_a^2\cdot\Var_\psi\left(\sum_{i\in J_a}
    \frac{\norm{H_i}}{c_a}\cdot\frac{H_i}{\norm{H_i}}\right)\\
  &\le c_a^2\cdot\LocalMeasurementVariance{\psi},
 \end{aligned}
\]
where the inequality uses \cref{def:local-measurement-variance}.

We combine the groups using the fact that $\sqrt{\Var_\psi(H)}$ is
a seminorm on Hermitian operators.  Indeed,
\[
 \begin{aligned}
  \norm{(H-\bra\psi H\ket\psi I)\ket\psi}^2
  &=\bra\psi\bigl(H-\bra\psi H\ket\psi I\bigr)^2\ket\psi\\
  &=\bra\psi H^2\ket\psi-\bigl(\bra\psi H\ket\psi\bigr)^2
   =\Var_\psi(H).
 \end{aligned}
\]
Applying this identity to the sums over the groups $J_a$, we obtain
\[
 \begin{aligned}
  \sqrt{\Var_\psi\left(\sum_i H_i\right)}
  &=\norm{\sum_a\sum_{i\in J_a}
     (H_i-\bra\psi H_i\ket\psi I)\ket\psi}\\
  &\le\sum_a\norm{\sum_{i\in J_a}
     (H_i-\bra\psi H_i\ket\psi I)\ket\psi}\\
  &=\sum_a\sqrt{\Var_\psi\left(\sum_{i\in J_a}H_i\right)}
   \le\sqrt{\LocalMeasurementVariance{\psi}}\cdot\sum_a c_a\\
  &\le\left(r\cdot\LocalMeasurementVariance{\psi}
     \cdot\sum_a c_a^2\right)^{1/2}
   =\left(r\cdot\LocalMeasurementVariance{\psi}
     \cdot\sum_i\norm{H_i}^2\right)^{1/2},
 \end{aligned}
\]
where the first inequality uses the triangle inequality for the vector
norm, the second uses the group bounds, and the third uses
Cauchy--Schwarz and the fact that there are at most $r$ groups.
Squaring proves the lemma.
\end{proof}

The next lemma controls the change caused by a single product reflection.

\begin{lemma}
\label{lem:product-reflection-change}
Let $R=I-2\Pi$ be a product reflection with support $G$, and let
$H_1,\ldots,H_h$ be Hermitian operators with $\norm{H_i}\le1$.
Suppose there are pairwise disjoint sets $A_1,\ldots,A_h\subseteq G$
such that $\supp(H_i)\cap G\subseteq A_i$ for $1\le i\le h$.
Then, for any real coefficients $c_1,\ldots,c_h$,
\[
 \norm{\sum_{i=1}^h c_i\cdot(RH_iR-H_i)}
 \le2\cdot\left(\sum_{i=1}^h c_i^2\right)^{1/2}.
\]
\end{lemma}

\begin{proof}
Let $\ket\theta$ be the product state on $G$ corresponding to $R$,
and let $\ket{\theta'}$ be any global pure state that agrees
with $\ket\theta$ on $G$.  For $A\subseteq G$, write
$\ket{\theta_A}$ for the factor of $\ket\theta$ on $A$.
The key observation is that the operators $H_i$ move
$\ket{\theta'}$ out of the range of $\Pi$ in mutually
orthogonal directions.

Set
$\ket{\phi_i}=(I-\Pi)H_i\ket{\theta'}$ and write
$\Pi_i=\proj{\theta_{A_i}}\otimes I_{A_i^c}$.
Since $H_i$ acts trivially on $G\setminus A_i$, we have
$\ket{\phi_i}=(I-\Pi_i)H_i\ket{\theta'}$.
For $i\ne j$, disjointness of $A_i$ and $A_j$ implies that
$\Pi_i$ and $H_j$ commute.  Hence,
\[
 \Pi_i\ket{\phi_j}
 =\Pi_i(I-\Pi)H_j\ket{\theta'}
 =(I-\Pi)H_j\Pi_i\ket{\theta'}
 =(I-\Pi)H_j\ket{\theta'}
 =\ket{\phi_j}.
\]
Therefore,
\[
 \langle\phi_i\mid\phi_j\rangle
 =\bra{\theta'} H_i(I-\Pi_i)\Pi_i\ket{\phi_j}
 =0.
\]
Write $H=\sum_{i=1}^h c_i\cdot H_i$.  By orthogonality,
\[
 \norm{(I-\Pi)H\ket{\theta'}}^2
 =\norm{\sum_{i=1}^h c_i\ket{\phi_i}}^2
 =\sum_{i=1}^h c_i^2\norm{\ket{\phi_i}}^2
 \le\sum_{i=1}^h c_i^2,
\]
where $\norm{\ket{\phi_i}}=\norm{(I-\Pi)H_i\ket{\theta'}}
\le\norm{H_i}\le1$.
The states $\ket{\theta'}$ range over all unit vectors in
$\operatorname{Im}(\Pi)$.  Since the rightmost projector $\Pi$
restricts the input to this subspace, we obtain
\[
 \norm{(I-\Pi)H\Pi}
 =\sup_{\ket v\in\operatorname{Im}(\Pi)}\norm{(I-\Pi)H\ket v}
 =\sup_{\ket{\theta'}}\norm{(I-\Pi)H\ket{\theta'}}
 \le\left(\sum_{i=1}^h c_i^2\right)^{1/2}.
\]

It remains to relate this to reflection by $R$.  Relative to
$\operatorname{Im}(\Pi)\oplus\operatorname{Im}(\Pi)^\perp$, write
\[
 H=\begin{pmatrix} A&B^\dagger\\ B&C\end{pmatrix},
 \qquad
 R=\begin{pmatrix}-I&0\\0&I\end{pmatrix},
\]
where $B$ is the block of $(I-\Pi)H\Pi$ mapping between the two
subspaces.  Hence,
\[
 \norm{RHR-H}
 =\norm{-2\begin{pmatrix}0&B^\dagger\\B&0\end{pmatrix}}
 =2\norm B
 \le2\cdot\left(\sum_{i=1}^h c_i^2\right)^{1/2},
\]
which is the desired bound.
\end{proof}

We now prove \cref{lem:variance-layer-growth}.  The idea is to group
the changes caused by the layer so that
\cref{lem:product-reflection-change} controls each group's size and
\cref{lem:overlap-variance} combines the groups into a variance bound.
The geometric part of the proof is to choose these groups so that
their supports have limited column overlap.

\lemlayergrow*

\begin{proof}[Proof of \cref{lem:variance-layer-growth}]
Fix connected output regions $D_1,\ldots,D_m$ with disjoint column
projections, Hermitian contractions $H_i$ supported on $D_i$, and
real coefficients with $\sum_i c_i^2=1$.  Set $H=\sum_i c_i\cdot H_i$.
We need to bound $\Var_\psi(U^\dagger HU)$.  We choose an order in
which to conjugate each $H_i$ by the gates meeting $D_i$ and analyze
the resulting sequence using a hybrid argument.  We use
\cref{lem:overlap-variance,lem:product-reflection-change} to control
the error terms arising in this argument.

For a nonempty set $X$, let $\ell(X)=\min\{c:c\in I(X)\}$ be its leftmost
column index.  For each region $D_i$, let $\mathcal L_i$ consist of the gates
meeting $D_i$ whose leftmost column is at or to the left of
$\ell(D_i)$, and let $\mathcal R_i$ consist of those starting to
its right.  Thus,
\[
 \mathcal L_i=\{g\in\mathcal G:D_i\cap G_g\ne\varnothing,\ 
                         \ell(G_g)\le\ell(D_i)\},
\]
\[
 \mathcal R_i=\{g\in\mathcal G:D_i\cap G_g\ne\varnothing,\ 
                         \ell(D_i)<\ell(G_g)\}.
\]
Similarly, viewed from gate $g$, let $\mathcal L_g$ contain the
regions meeting $G_g$ that start to its left, and let $\mathcal R_g$
contain those starting at or to its right.  Equivalently,
$\mathcal L_g=\{i:g\in\mathcal R_i\}$ and
$\mathcal R_g=\{i:g\in\mathcal L_i\}$.

\begin{figure}[htbp]
\centering
\begin{tikzpicture}[
  x=0.8cm,y=0.8cm,font=\small,
  gate/.style={draw=black!85,line width=1pt},
  gate label/.style={text=black,inner sep=2pt},
  qubit/.style={circle,fill=black!85,inner sep=0pt,minimum size=3.5pt}
]
  \fill[teal!15] (2,0) rectangle (4,2);
  \fill[teal!15] (4,1) rectangle (8,2);
  \fill[violet!13] (7,2) rectangle (9,3);
  \fill[violet!13] (8,1) rectangle (9,2);
  \fill[violet!13] (9,2) rectangle (12,3);
  \fill[teal!15] (7,1) rectangle (8,2);
  \fill[violet!20,opacity=0.6] (7,1) rectangle (8,2);

  \draw[step=1,black!22,line width=0.35pt] (0,0) grid (12,4);
  \draw[black!45,line width=0.6pt] (0,0) rectangle (12,4);
  \draw[teal!75!black,line width=1pt]
    (2,0)--(4,0)--(4,1)--(8,1)--(8,2)--(2,2)--cycle;
  \draw[violet!75!black,line width=1pt]
    (7,1)--(9,1)--(9,2)--(12,2)--(12,3)--(7,3)--cycle;
  \draw[teal!75!black,densely dashed,line width=1.1pt]
    (2.08,0.08) rectangle (3.92,1.92);
  \draw[violet!75!black,densely dashed,line width=1.1pt]
    (7.08,1.08) rectangle (8.92,2.92);

  \draw[gate] (0.5,0.5)--(3.5,0.5);
  \foreach \x in {0.5,1.5,2.5,3.5}
    \node[qubit] at (\x,0.5) {};
  \node[gate label,above=3pt] at (1,0.5) {$G_a$};
  \draw[gate] (3.5,1.5)--(7.5,1.5);
  \foreach \x in {3.5,4.5,5.5,6.5,7.5}
    \node[qubit] at (\x,1.5) {};
  \node[gate label,above=3pt] at (5.5,2) {$G_b$};
  \draw[gate] (8.5,2.5)--(11.5,2.5);
  \foreach \x in {8.5,9.5,10.5,11.5}
    \node[qubit] at (\x,2.5) {};
  \node[gate label,above=3pt] at (10.5,3) {$G_c$};

  \node[text=teal!75!black] at (3,4.55) {$D_1$};
  \draw[teal!75!black] (3,4.25)--(3,2.12);
  \node[text=violet!75!black] at (8,4.55) {$D_2$};
  \draw[violet!75!black] (8,4.25)--(8,3.12);
  \node[text=teal!75!black] at (3,-0.65)
    {$\widetilde D_1=D_1\cup G_b$};
  \node[text=violet!75!black] at (9,-0.65)
    {$\widetilde D_2=D_2\cup G_c$};
\end{tikzpicture}
\caption{Gate supports from left to right: $G_a\to G_b\to G_c$.}
\label{fig:variance-region-absorption}
\end{figure}

Let $k=|\mathcal G|$ be the number of gates in the layer.
List these gates as $g_1,\ldots,g_k$ in order of their
leftmost columns, from right to left, breaking ties arbitrarily.  Since the
gates commute, we may use this order to compute $U^\dagger HU$.
For each $i$, every gate in $\mathcal R_i$ comes before every gate
in $\mathcal L_i$.  Let $s_i$ be the index of the last gate in
$\mathcal R_i$, with $s_i=0$ if $\mathcal R_i$ is empty.
We begin tracking $H_i$ at step $s_i$, after all gates in
$\mathcal R_i$ have been applied.  At that point its support is
contained in
\[
 \widetilde D_i=D_i\cup\bigcup_{g\in\mathcal R_i}G_g.
\]
See \cref{fig:variance-region-absorption} for a detailed illustration
of this procedure.
This set is connected because every added gate meets $D_i$.
The observable at this starting point is the Hermitian contraction
supported on $\widetilde D_i$ given by
\[
 \widetilde H_i=\widetilde H_i^{(s_i)}
 =\left(\prod_{g\in\mathcal R_i}R_g\right)^\dagger
 H_i\left(\prod_{g\in\mathcal R_i}R_g\right).
\]
\Cref{fig:hybrid-starting-times} shows how these starting points
differ between observables within the same global gate order.

\begin{figure}[htbp]
\centering
\begin{tikzpicture}[
  x=0.7cm,y=0.7cm,font=\small,
  gate/.style={black,line width=1pt},
  qubit/.style={circle,fill=black,inner sep=0pt,minimum size=3pt},
  start/.style={circle,inner sep=0pt,minimum size=6pt},
  update/.style={circle,fill=white,inner sep=0pt,minimum size=6pt,line width=1pt}
]
  \fill[teal!15] (2,0) rectangle (4,2);
  \fill[violet!13] (6,1) rectangle (8,4);
  \fill[blue!12] (10,2) rectangle (12,5);
  \draw[step=1,black!22,line width=0.35pt] (0,0) grid (14,5);
  \draw[black!45,line width=0.6pt] (0,0) rectangle (14,5);
  \draw[teal!75!black,line width=1pt] (2,0) rectangle (4,2);
  \draw[violet!75!black,line width=1pt] (6,1) rectangle (8,4);
  \draw[blue!70!black,line width=1pt] (10,2) rectangle (12,5);
  \node[text=teal!75!black] at (3,5.55) {$D_1$};
  \draw[teal!75!black] (3,5.25)--(3,2.12);
  \node[text=violet!75!black] at (7,5.55) {$D_2$};
  \draw[violet!75!black] (7,5.25)--(7,4.12);
  \node[text=blue!70!black] at (11,5.55) {$D_3$};
  \draw[blue!70!black] (11,5.25)--(11,5.12);

  \draw[gate] (0.5,0.5)--(3.5,0.5);
  \foreach \x in {0.5,1.5,2.5,3.5}
    \node[qubit] at (\x,0.5) {};
  \node[above=3pt] at (1,0.5) {$G_a$};
  \draw[gate] (3.5,1.5)--(7.5,1.5);
  \foreach \x in {3.5,4.5,5.5,6.5,7.5}
    \node[qubit] at (\x,1.5) {};
  \node[above=3pt] at (5,1.5) {$G_b$};
  \draw[gate] (7.5,2.5)--(11.5,2.5);
  \foreach \x in {7.5,8.5,9.5,10.5,11.5}
    \node[qubit] at (\x,2.5) {};
  \node[above=3pt] at (9,2.5) {$G_c$};
  \draw[gate] (11.5,3.5)--(13.5,3.5);
  \foreach \x in {11.5,12.5,13.5}
    \node[qubit] at (\x,3.5) {};
  \node[above=3pt] at (13,3.5) {$G_e$};
  \draw[gate] (7.5,3.5)--(9.5,3.5);
  \foreach \x in {7.5,8.5,9.5}
    \node[qubit] at (\x,3.5) {};
  \node[above=3pt] at (9,3.5) {$G_d$};

  \node[anchor=east] at (2.6,-1.15) {Global gate order};
  \draw[->,>=Stealth,black!65] (3,-1.15)--(13.8,-1.15);
  \foreach \x/\gate in {4/e,6.25/d,8.5/c,10.75/b,13/a} {
    \node[above=4pt] at (\x,-1.15) {$\gate$};
    \draw[black!20,densely dotted] (\x,-1.4)--(\x,-4.55);
  }

  \node[anchor=east,text=teal!75!black] at (0.7,-2.25) {$H_1$};
  \draw[teal!75!black,dashed] (1,-2.25)--(10.75,-2.25);
  \draw[teal!75!black,line width=1pt,->,>=Stealth] (10.75,-2.25)--(13.8,-2.25);
  \node[start,fill=teal!75!black] (startH1) at (10.75,-2.25) {};
  \node[update,draw=teal!75!black] at (13,-2.25) {};

  \node[anchor=east,text=violet!75!black] at (0.7,-3.25) {$H_2$};
  \draw[violet!75!black,dashed] (1,-3.25)--(8.5,-3.25);
  \draw[violet!75!black,line width=1pt,->,>=Stealth] (8.5,-3.25)--(13.8,-3.25);
  \node[start,fill=violet!75!black] (startH2) at (8.5,-3.25) {};
  \node[update,draw=violet!75!black] at (10.75,-3.25) {};
  \node[update,draw=violet!75!black] at (13,-3.25) {};

  \node[anchor=east,text=blue!70!black] at (0.7,-4.25) {$H_3$};
  \draw[blue!70!black,dashed] (1,-4.25)--(4,-4.25);
  \draw[blue!70!black,line width=1pt,->,>=Stealth] (4,-4.25)--(13.8,-4.25);
  \node[start,fill=blue!70!black] (startH3) at (4,-4.25) {};
  \node[update,draw=blue!70!black] at (8.5,-4.25) {};
  \node[update,draw=blue!70!black] at (6.25,-4.25) {};
  \node[update,draw=blue!70!black] at (10.75,-4.25) {};
  \node[update,draw=blue!70!black] at (13,-4.25) {};

  \node[anchor=north,text=blue!70!black] (labelH3) at (4,-5.1) {$\widetilde H_3^{(s_3)}$};
  \node[anchor=north,text=violet!75!black] (labelH2) at (8.5,-5.1) {$\widetilde H_2^{(s_2)}$};
  \node[anchor=north,text=teal!75!black] (labelH1) at (10.75,-5.1) {$\widetilde H_1^{(s_1)}$};
  \draw[blue!70!black,line width=0.5pt] (labelH3.north)--(startH3.south);
  \draw[violet!75!black,line width=0.5pt]
    (labelH2.north) .. controls (7.9,-4.65) and (7.9,-3.8) .. (startH2.south west);
  \draw[teal!75!black,line width=0.5pt]
    (labelH1.north) .. controls (9.55,-4.55) and (9.55,-2.9) .. (startH1.south west);
\end{tikzpicture}
\caption{Different starting times within one global gate order:
$s_1=4$, $s_2=3$, and $s_3=1$.}
\label{fig:hybrid-starting-times}
\end{figure}

Let $\mathcal L=\bigcup_i\mathcal L_i$.  For each $H_i$, the remaining
changes come from the gates in $\mathcal L_i$, all of which have
indices greater than $s_i$.  We use the global gate index $j$ and
define the sequence for $j=s_i,\ldots,k$.  For $j>s_i$, set
\[
 \widetilde H_i^{(j)}=
 \begin{cases}
  R_{g_j}^\dagger\,\widetilde H_i^{(j-1)}\,R_{g_j}, & \quad g_j\in\mathcal L_i,\\[4pt]
  \widetilde H_i^{(j-1)}, & \quad g_j\notin\mathcal L_i.
 \end{cases}
\]
Thus, each observable is conjugated by exactly the gates in its
own $\mathcal L_i$, giving
\[
 \widetilde H_i^{(k)}
 =\left(\prod_{g\in\mathcal L_i}R_g\right)^\dagger
     \widetilde H_i\left(\prod_{g\in\mathcal L_i}R_g\right)
 =U^\dagger H_iU,
\]
because the gates in the layer commute.

For each global step $j$, every $i\in\mathcal R_{g_j}$ satisfies
$s_i<j$.  Let $\Delta_j$ collect the changes to these observables:
\[
 \Delta_j=\sum_{i\in\mathcal R_{g_j}}
        c_i\cdot\bigl(R_{g_j}\widetilde H_i^{(j-1)}R_{g_j}
                     -\widetilde H_i^{(j-1)}\bigr).
\]
Write $\widetilde H=\sum_i c_i\cdot\widetilde H_i$.
Telescoping each sequence from its own starting index $s_i$ gives
\[
 \begin{aligned}
  U^\dagger HU
  &=\sum_i c_i\cdot\widetilde H_i^{(k)}\\
  &=\sum_i c_i\cdot\widetilde H_i
    +\sum_i\sum_{j=s_i+1}^k c_i\cdot
       \bigl(\widetilde H_i^{(j)}-\widetilde H_i^{(j-1)}\bigr)\\
  &=\widetilde H+\sum_{j=1}^k\Delta_j.
 \end{aligned}
\]

We first bound the variance of $\widetilde H$ by counting how many enlarged
regions meet a column $c$.  Every region in $\mathcal L_g$ meets the
leftmost column of $G_g$.  Since no two regions meet the same column,
\[
 |\mathcal L_g|
 \le\#\{i:\ell(G_g)\in I(D_i)\}\le1.
\]
Thus, each gate is added to at most one region.  At most one original
region meets $c$, and at most $w$ gates meet $c$, so
\[
 \begin{aligned}
  \#\{i:c\in I(\widetilde D_i)\}
  &\le\#\{i:c\in I(D_i)\}
      +\sum_i\#\{g\in\mathcal R_i:c\in I(G_g)\}\\
  &\le1+w.
 \end{aligned}
\]
Applying
\cref{lem:overlap-variance} to the summands, we obtain
\[
 \begin{aligned}
  \Var_\psi(\widetilde H)
  &\le(w+1)\cdot\LocalMeasurementVariance{\psi}
       \cdot\sum_i\norm{c_i\cdot\widetilde H_i}^2\\
  &=(w+1)\cdot\LocalMeasurementVariance{\psi}
       \cdot\sum_i c_i^2\cdot\norm{\widetilde H_i}^2\\
  &\le(w+1)\cdot\LocalMeasurementVariance{\psi}
       \cdot\sum_i c_i^2
   =(w+1)\cdot\LocalMeasurementVariance{\psi}.
 \end{aligned}
\]
Here we used $\norm{\widetilde H_i}\le1$ and $\sum_i c_i^2=1$.

We use \cref{lem:overlap-variance} to bound the variance of
$\sum_{j=1}^k\Delta_j$.  We therefore need connected sets supporting
the $\Delta_j$ and bounds on both their column overlap and the sum
of squared norms $\sum_{j=1}^k\norm{\Delta_j}^2$.

\textbf{Support overlap.}
If $g_j\notin\mathcal L$, then
$\Delta_j=0$.  For $g_j\in\mathcal L$, take the current
gate, the enlarged regions it meets as a remaining gate, and all
earlier gates applied to those observables:
\[
 \widetilde G_{g_j}=G_{g_j}
   \cup\bigcup_{i\in\mathcal R_{g_j}}\widetilde D_i
   \cup\bigcup_{\substack{a<j:\ \exists i\\
                      g_a,g_j\in\mathcal L_i}}G_{g_a}.
\]
Every region in this union meets $G_{g_j}$, and every earlier gate
in the union meets one of those regions.  Thus, $\widetilde G_{g_j}$ is connected
and contains the support of $\Delta_j$.
\Cref{fig:hybrid-gate-support} illustrates this set for gate $b$
after adding one gate to the configuration of \cref{fig:hybrid-starting-times}.

\begin{figure}[htbp]
\centering
\begin{tikzpicture}[
  x=0.7cm,y=0.7cm,font=\small,
  gate/.style={black,line width=1pt},
  qubit/.style={circle,fill=black,inner sep=0pt,minimum size=3pt},
  inactive gate/.style={black!35,line width=1pt},
  inactive qubit/.style={qubit,fill=black!35}
]
  \draw[teal!45!black,line width=0.7pt] (2,0) rectangle (4,2);
  \draw[blue!45!black,line width=0.7pt] (10,2) rectangle (12,5);
  \fill[violet!15]
    (3,1)--(8,1)--(8,2)--(12,2)--(12,3)--(10,3)
    --(10,4)--(6,4)--(6,3)--(4,3)--(4,2)--(3,2)--cycle;
  \draw[step=1,black!22,line width=0.35pt] (0,0) grid (14,5);
  \draw[black!45,line width=0.6pt] (0,0) rectangle (14,5);
  \draw[violet!45,densely dashed,line width=1pt]
    (6.06,1.06) rectangle (7.94,3.94);
  \draw[violet!45,line width=1pt]
    (3,1)--(8,1)--(8,2)--(12,2)--(12,3)--(10,3)
    --(10,4)--(6,4)--(6,3)--(4,3)--(4,2)--(3,2)--cycle;
  \node[text=teal!45!black] at (3,5.55) {$D_1$};
  \draw[teal!45!black] (3,5.25)--(3,2.12);
  \node[text=violet!75!black] at (7,5.55) {$D_2$};
  \draw[violet!45] (7,5.25)--(7,4.12);
  \node[text=blue!45!black] at (11,5.55) {$D_3$};
  \draw[blue!45!black] (11,5.25)--(11,5.12);

  \draw[inactive gate] (0.5,0.5)--(3.5,0.5);
  \foreach \x in {0.5,1.5,2.5,3.5}
    \node[inactive qubit] at (\x,0.5) {};
  \node[above=3pt,text=black!45] at (1,0.5) {$G_a$};
  \draw[gate] (3.5,1.5)--(7.5,1.5);
  \foreach \x in {3.5,4.5,5.5,6.5,7.5}
    \node[qubit] at (\x,1.5) {};
  \node[above=3pt] at (5,1.5) {$G_b$};
  \draw[gate] (4.5,2.5)--(6.5,2.5);
  \foreach \x in {4.5,5.5,6.5}
    \node[qubit] at (\x,2.5) {};
  \node at (5,3.5) {$G_c$};
  \draw[gate] (7.5,2.5)--(11.5,2.5);
  \foreach \x in {7.5,8.5,9.5,10.5,11.5}
    \node[qubit] at (\x,2.5) {};
  \node[above=3pt] at (9,2.5) {$G_d$};
  \draw[gate] (7.5,3.5)--(9.5,3.5);
  \foreach \x in {7.5,8.5,9.5}
    \node[qubit] at (\x,3.5) {};
  \node[above=3pt] at (9,3.5) {$G_e$};
  \draw[inactive gate] (11.5,3.5)--(13.5,3.5);
  \foreach \x in {11.5,12.5,13.5}
    \node[inactive qubit] at (\x,3.5) {};
  \node[above=3pt,text=black!45] at (13,3.5) {$G_f$};

  \node[text=violet!75!black] at (7,-0.75)
    {$\widetilde G_b=G_b\cup\widetilde D_2\cup G_c
      =G_b\cup D_2\cup G_d\cup G_e\cup G_c$};
\end{tikzpicture}
\caption{Global gate order: $f\to e\to d\to c\to b\to a$.}
\label{fig:hybrid-gate-support}
\end{figure}

Every gate in $\mathcal L_i$ meets the leftmost column of $D_i$, so
\[
 |\mathcal L_i|
 \le\#\{g:\ell(D_i)\in I(G_g)\}\le w.
\]

We also collect the indices of gates that appear after $g_j$ in the
global order and whose supports share a column with $G_{g_j}$.  Let
$\mathcal R_j^{\mathrm{gate}}
=\{a>j:I(G_{g_a})\cap I(G_{g_j})\ne\varnothing\}$.
If a gate appearing later in the list shares a column with $G_{g_j}$, it must
meet the leftmost column of $G_{g_j}$.  Since $g_j$ itself also meets
that column, there are at most $w-1$ such later gates:
\[
 \begin{aligned}
  |\mathcal R_j^{\mathrm{gate}}|
  &\le\#\{a>j:\ell(G_{g_j})\in I(G_{g_a})\}\\
  &\le w-1.
 \end{aligned}
\]

For each column $c$,
\[
 \begin{aligned}
  \#\{g\in\mathcal L:c\in I(\widetilde G_g)\}
  &\le\sum_{i:c\in I(\widetilde D_i)}|\mathcal L_i|
       +\sum_{a:c\in I(G_{g_a})}
          \bigl(1+|\mathcal R_a^{\mathrm{gate}}|\bigr)\\
  &\le(w+1)\cdot w+w+w\cdot(w-1)=2w^2.
 \end{aligned}
\]
The two sums count contributions from enlarged regions and gates,
respectively.  Each gate $g_a$ contributes through itself and through
later gates $g_j$ sharing some $\mathcal L_i$; such indices $j$ lie in
$\mathcal R_a^{\mathrm{gate}}$ because both gates meet the leftmost
column of $D_i$.

\textbf{Norm bounds.}
Before applying $g_j$, the
support of an observable can have grown only along gates already
applied to it.  All of those gates are disjoint from $G_{g_j}$, so
for $g_j\in\mathcal L_i$,
\[
 \begin{aligned}
  \supp(\widetilde H_i^{(j-1)})\cap G_{g_j}
  &\subseteq\left(D_i\cup\bigcup_{g\in\mathcal R_i}G_g
     \cup\bigcup_{\substack{a<j\\g_a\in\mathcal L_i}}G_{g_a}\right)\cap G_{g_j}\\
  &=D_i\cap G_{g_j}.
 \end{aligned}
\]
The sets $D_i\cap G_{g_j}$ are disjoint.  Thus,
\cref{lem:product-reflection-change} applies to the observables
$\widetilde H_i^{(j-1)}$ changed by this gate, with
$A_i=D_i\cap G_{g_j}$ and coefficients $c_i$.

By \cref{lem:product-reflection-change}, this Hermitian operator satisfies
\[
 \norm{\Delta_j}^2
 \le\left(2\cdot
     \left(\sum_{i\in\mathcal R_{g_j}}c_i^2\right)^{1/2}\right)^2
 =4\cdot\sum_{i\in\mathcal R_{g_j}}c_i^2.
\]
Summing over the global steps gives
\[
 \begin{aligned}
  \sum_{j=1}^k\norm{\Delta_j}^2
  &\le4\cdot\sum_{g\in\mathcal L}\sum_{i:g\in\mathcal L_i}c_i^2
   =4\cdot\sum_i|\mathcal L_i|\cdot c_i^2\\
  &\le4w\cdot\sum_i c_i^2=4w.
 \end{aligned}
\]
Here each $c_i^2$ is counted at most $|\mathcal L_i|\le w$ times.
Applying \cref{lem:overlap-variance} to the $\Delta_j$ gives
\[
 \begin{aligned}
  \Var_\psi\left(\sum_{j=1}^k\Delta_j\right)
  &\le2w^2\cdot\LocalMeasurementVariance{\psi}
       \cdot\sum_{j=1}^k\norm{\Delta_j}^2\\
  &\le2w^2\cdot\LocalMeasurementVariance{\psi}\cdot4w
   =8w^3\cdot\LocalMeasurementVariance{\psi},
 \end{aligned}
\]
using the column count $2w^2$ and
$\sum_{j=1}^k\norm{\Delta_j}^2\le4w$.

Combining the variance bounds for $\widetilde H$ and
$\sum_{j=1}^k\Delta_j$ with the decomposition above, we obtain
\[
 \begin{aligned}
  \Var_{U\psi}(H)
  &=\Var_\psi(U^\dagger HU)
   =\Var_\psi\left(\widetilde H+\sum_{j=1}^k\Delta_j\right)\\
  &\le\left(\sqrt{\Var_\psi(\widetilde H)}
      +\sqrt{\Var_\psi\left(\sum_{j=1}^k\Delta_j\right)}\right)^2\\
  &\le2\cdot\Var_\psi(\widetilde H)
      +2\cdot\Var_\psi\left(\sum_{j=1}^k\Delta_j\right)\\
  &\le\bigl(2(w+1)+16w^3\bigr)\cdot\LocalMeasurementVariance{\psi}
   \le O\left(w^3\right)\cdot\LocalMeasurementVariance{\psi},
 \end{aligned}
\]
where we use the triangle inequality for the square root of variance
from the proof of \cref{lem:overlap-variance}.
Taking the maximum in the definition of
$\LocalMeasurementVariance{U\psi}$ proves the lemma.
\end{proof}

\subsection{Fourier property of \texorpdfstring{$\twoDQAC$}{2D-QAC} functions}
\label{sec:fourier-concentration}

For Boolean functions computed exactly by $\twoDQAC$ circuits,
\cref{thm:fourier-tail} gives a Fourier tail bound, and
\cref{thm:boolean-sensitivity-bound} bounds the average sensitivity
$\AverageSensitivity(f)$. Both arguments encode the computed value
as a phase and apply the output variance bound. We use the Boolean
and Fourier notation from \cref{sec:boolean-preliminaries}.

Arunachalam, Bravyi, Dutt, and Yoder
\cite[Eqs.~(1) and~(3)]{arunachalam2022learningphasestates} define binary
phase states with amplitudes $2^{-n/2}(-1)^{f(x)}$ and generalized
phase states with amplitudes $2^{-n/2}e^{2\pi i f(x)/q}$.
We consider the $i^f$ encoding, the $q=4$ case of the latter
form with Boolean $f$, and denote it by
\[
 \ket{\chi_f}=2^{-n/2}\sum_x i^{f(x)}\ket x.
\]
Suppose a depth-$d$ circuit $C$ computes $f$ exactly, and let
$D=\diag(1,i)$ act on its output qubit. Since that qubit is in
$\ket{f(x)}$, we have $C^\dagger DC\ket{x,0^a}=i^{f(x)}\ket{x,0^a}$.
Thus
\begin{equation}
 C^\dagger DC\left(\ket+^{\otimes n}\ket{0^a}\right)
 =\ket{\chi_f}\ket{0^a}.
 \label{eq:phase-uncompute}
\end{equation}
The resulting circuit has depth $2d$ on the same grid. A lower bound
on the variance of $\ket{\chi_f}$ therefore constrains the depth and
width needed to compute $f$.

For the Fourier tail bound, we apply \cref{cor:additive-variance-bound}
at depth $2d$ to control the variance of the sampled degree. The empty
sample has probability at least $1/2$, so Cauchy--Schwarz converts
this into a second-moment bound.

We consider fixed $d$, so the constants in the $O$ notation below
may depend on $d$.
\begin{theorem}
\label{thm:fourier-tail}
If a depth-$d$, width-$w$ $\twoDQAC$ circuit computes $f :  \set{0,1}^n \to \set{0,1}$ exactly,
then, for every integer $k\ge1$ and constant $d \geq 1$,
\[
 \FourierWeight{\ge k}{h}
 \le O\left(\frac{n w^{6d+1}}{(1+\mu^2) k^2}\right).
\]
Here $h(x)=(-1)^{f(x)}$ is the sign representation of $f$, and
$\mu=\E_x \Br{h(x)}$.
\end{theorem}

\begin{proof}
Let $g = i^f : \set{0,1}^n \to \mathbb{C}$. Apply Hadamards to the input register after \eqref{eq:phase-uncompute}.
Measuring this register samples $S$ with probability $|\widehat g(S)|^2$.
This circuit still has depth $2d$ on the same grid.

From $g(x)=(1+i)/2+(1-i) h(x)/2$, the sampling probabilities are
\[
 \Pr_g(S)=
 \begin{cases}
  \dfrac{1+\mu^2}{2}, & \quad S=\varnothing,\\[6pt]
  \dfrac{\widehat h(S)^2}{2}, & \quad S\ne\varnothing.
 \end{cases}
\]
The sampled degree $|S|$ is equivalently obtained by measuring
$\frac12(nI-\sum_jX_j)$ on $\ket{\chi_f}$ before the Hadamards.
Applying \cref{cor:additive-variance-bound} at depth $2d$ therefore gives
\[
 \begin{aligned}
  \Var_g(|S|)
  &=\frac14 \Var_{\chi_f}\left(\sum_jX_j\right)
  \le\frac14 K^{2d} \cdot n w^{6d+1}.
 \end{aligned}
\]
Since $|S|=0$ on the empty sample, Cauchy--Schwarz gives
\[
 \bigl(\E_g|S|\bigr)^2
 \le\Pr_g[S\ne\varnothing] \cdot \E_g \Br{|S|^2}
 = \br{\frac{1-\mu^2}{2}} \cdot \E_g \Br{|S|^2}.
\]
Thus $\E_g \Br{|S|^2}\le2 \Var_g(|S|)/(1+\mu^2)$, and
\[
 \begin{aligned}
  k^2 \FourierWeight{\ge k}{h}
  &\le\sum_{S\subseteq[n]}|S|^2 \widehat h(S)^2
   =2 \E_g \Br{|S|^2}\\
  &\le\frac{4 \Var_g(|S|)}{1+\mu^2}
  \le\frac{K^{2d} \cdot n w^{6d+1}}{1+\mu^2}.
 \end{aligned}
\]
\end{proof}

Combining the phase-state variance with \cref{cor:additive-variance-bound}
also bounds average sensitivity, with the same constant $K$.
\begin{theorem}
\label{thm:boolean-sensitivity-bound}
Suppose a depth-$d$, width-$w$ $\twoDQAC$ circuit computes
$f:\bits^n\to\bits$ exactly from product ancillas. Then
$\AverageSensitivity(f)\le K^d\cdot\sqrt n\cdot w^{3d+1/2}$.
For $d\ge1$, this implies
\[
 w=\Omega\left(
       \left(\frac{\AverageSensitivity(f)^2}{n}\right)^{1/(6d+1)}
     \right).
\]
The implicit constant is absolute and independent of $d$.
\end{theorem}

The following lemma relates phase variance to the first two moments
of sensitivity.

\begin{lemma}
\label{lem:phase-sensitivity-variance}
For every Boolean function $f:\bits^n\to\bits$, the phase state
$\ket{\chi_f}$ satisfies
\[
 \Var_{\chi_f}\left(\sum_jX_j\right)
 =2\cdot\E_x[s_f(x)^2]-\AverageSensitivity(f)^2.
\]
As a corollary, on a width-$w$ grid,
\[
 \LocalMeasurementVariance{\chi_f}
 \ge\frac{2\cdot\E_x[s_f(x)^2]-\AverageSensitivity(f)^2}{nw}
 \ge\frac{\AverageSensitivity(f)^2}{nw}.
\]
\end{lemma}

\begin{proof}
For each coordinate $j$, a bit flip changes the phase exactly when
it changes $f$, so
$\norm{(I-X_j)\ket{\chi_f}}^2=2\Pr_x[f(x)\ne f(x\oplus e_j)]$.
Using $(I-X_j)^2=2(I-X_j)$ and summing over $j$ gives the first
moment below. For the second moment, the identity
$g=(1+i)/2+(1-i)h/2$ gives
\[
 \left(nI-\sum_jX_j\right)\ket{\chi_f}
 =2^{-n/2}(1-i)\sum_x s_f(x)h(x)\ket x.
\]
Consequently,
\[
 \begin{aligned}
  \bra{\chi_f}\left(nI-\sum_jX_j\right)\ket{\chi_f}
  &=\AverageSensitivity(f),\\
  \bra{\chi_f}\left(nI-\sum_jX_j\right)^2\ket{\chi_f}
  &=2\E_x[s_f(x)^2].
 \end{aligned}
\]
Subtracting the square of the first moment gives the variance
identity. Grouping the observables by column as in the proof of
\cref{prop:dicke-variance} gives
$\Var_{\chi_f}(\sum_jX_j)\le nw\LocalMeasurementVariance{\chi_f}$.
Cauchy--Schwarz then gives the lower bound in terms of average sensitivity.
\end{proof}

We now prove \cref{thm:boolean-sensitivity-bound}.

\begin{proof}
After absorbing the product ancillary initialization into local gates,
\eqref{eq:phase-uncompute} synthesizes $\ket{\chi_f}\ket{0^a}$ at depth
$2d$. Applying \cref{cor:additive-variance-bound} and
\cref{lem:phase-sensitivity-variance} gives
\[
 \AverageSensitivity(f)^2
 \le\Var_{\chi_f}\left(\sum_jX_j\right)
 \le K^{2d}\cdot n\cdot w^{6d+1}.
\]
Rearranging gives both bounds, with an absolute implicit constant.
\end{proof}

The phase-state representation also leads to a learning application.
In Appendix~\ref{sec:mps-learning}, we bound the Schmidt ranks across
column cuts and obtain a matrix product state representation. Combined
with tensor-network capacity and PAC learning bounds
\cite{khavari2021pseudodimension,hanneke2016optimal}, this gives a
learner using polynomially many uniform classical labeled examples
for functions computed exactly from product ancillas by constant-depth
$\twoDQAC$ circuits of logarithmic width and polynomial length on a
known layout.
The guarantee concerns sample complexity, not running time.

\subsection{Lower bound for \texorpdfstring{$\twoDQAC$}{2D-QAC} functions}
\label{sec:variance-applications}

This subsection gives width lower bounds for $\Parity$,
$\mathrm{MOD}$, and $\mathrm{MAJORITY}$ in
\cref{thm:parity-sensitivity}, \cref{thm:mod-sensitivity}, and
\cref{thm:majority-sensitivity}, respectively. All circuit claims in this
subsection concern exact computation in the $\twoDQAC$ model, with
product ancillas. The proofs use the average sensitivity bound in
\cref{thm:boolean-sensitivity-bound} and the sensitivity--variance
identity in \cref{lem:phase-sensitivity-variance}.

For $\Parity$ and $\mathrm{MOD}$, we compute the average sensitivity
and apply \cref{thm:boolean-sensitivity-bound}.

\begin{theorem}
\label{thm:parity-sensitivity}
The parity function satisfies $\AverageSensitivity(\Parity_n)=n$.
For fixed depth $d\ge1$, exact computation of $\Parity_n$ requires
\[
 w=\Omega(n^{1/(6d+1)}).
\]
\end{theorem}

\begin{proof}
Flipping any bit changes parity, so $s_{\Parity_n}(x)=n$ on every
input. The width bound follows from \cref{thm:boolean-sensitivity-bound}.
\end{proof}

\begin{remark}
For exact clean unitary synthesis of $\Parity_n$, the width bound improves to
$w=\Omega(n^{1/(3d+1)})$. Indeed, the parity--fanout conversion in
\cref{prop:parity-fanout-reductions}, proved in
\cref{sec:cat-fanout-exactification}, turns a depth-$d$ clean parity
circuit into an $(n+1)$-qubit cat-state synthesis circuit with the same depth
and qubit layout. Applying \cref{thm:width-lower-bound} thus gives the claim.
\end{remark}

For a fixed modulus $r\ge2$, let
$\mathrm{MOD}_{r,n}(x)=\mathbf1\{|x|\equiv0\,(\mathrm{mod}\,r)\}$.
A bit is pivotal when the other bits have weight congruent to zero
or $r-1$ modulo $r$. These two residue classes give linear average
sensitivity, to which \cref{thm:boolean-sensitivity-bound} applies.

\begin{theorem}
\label{thm:mod-sensitivity}
For every fixed $r\ge3$, $\AverageSensitivity(\mathrm{MOD}_{r,n})
=(2/r+o(1))\cdot n$.
For fixed depth $d\ge1$, exact computation of $\mathrm{MOD}_{r,n}$
requires
\[
 w=\Omega\left(\left(\frac{n}{r^2}\right)^{1/(6d+1)}\right).
\]
The implicit constant is absolute.
\end{theorem}

\begin{proof}
Let $B$ be binomial with parameters $n-1$ and $1/2$, and set
$\omega=e^{2\pi i/r}$. For each residue $u$, the roots-of-unity
identity gives
\[
 \Pr[B\equiv u\,(\mathrm{mod}\,r)]
 =\frac1r\sum_{\ell=0}^{r-1}\omega^{-u\ell}\cdot
       \left(\frac{1+\omega^\ell}{2}\right)^{n-1}
 =\frac1r+o(1).
\]
Summing the pivotal probabilities over the $n$ coordinates proves
\[
 \AverageSensitivity(\mathrm{MOD}_{r,n})
 =n\cdot\bigl(\Pr[B\equiv0\,(\mathrm{mod}\,r)]
             +\Pr[B\equiv r-1\,(\mathrm{mod}\,r)]\bigr)
 =\left(\frac2r+o(1)\right)\cdot n.
\]
For every fixed $r$ and sufficiently large $n$, this gives
$\AverageSensitivity(\mathrm{MOD}_{r,n})\ge n/r$.
Substituting into \cref{thm:boolean-sensitivity-bound} proves the width bound.
\end{proof}

We define $\mathrm{MAJ}_n(x)=\mathbf1\{|x|>n/2\}$.
For MAJORITY, \cref{thm:boolean-sensitivity-bound} alone gives no
growing width lower bound. The proof uses the first two moments of
sensitivity through \cref{lem:phase-sensitivity-variance}.

\begin{theorem}
\label{thm:majority-sensitivity}
For fixed depth $d\ge1$, exact computation of $\mathrm{MAJ}_n$
requires
\[
 w=\Omega(n^{1/(12d+2)}).
\]
\end{theorem}

\begin{proof}
Put $t=\lfloor n/2\rfloor$. A coordinate is pivotal exactly when
the other $n-1$ bits have weight $t$, so
$\AverageSensitivity(\mathrm{MAJ}_n)=n\cdot\binom{n-1}{t}/2^{n-1}=\Theta(\sqrt n)$.
Only inputs of weights $t$ and $t+1$ have nonzero
sensitivity, with values $n-t$ and $t+1$, respectively. The identity
$(n-t)\cdot\binom nt=(t+1)\cdot\binom n{t+1}$ therefore gives
\[
 \begin{aligned}
  \E_x[s_{\mathrm{MAJ}_n}(x)^2]
  &=2^{-n}\cdot\left((n-t)^2\cdot\binom nt
       +(t+1)^2\cdot\binom n{t+1}\right)\\
  &=\frac{n+1}{2}\cdot\AverageSensitivity(\mathrm{MAJ}_n)
   =\Theta(n^{3/2}).
 \end{aligned}
\]
Substituting into the variance identity in
\cref{lem:phase-sensitivity-variance} yields
\[
 \Var_{\chi_{\mathrm{MAJ}_n}}\left(\sum_jX_j\right)
 =(n+1)\cdot\AverageSensitivity(\mathrm{MAJ}_n)
   -\AverageSensitivity(\mathrm{MAJ}_n)^2
 =\Theta(n^{3/2}).
\]
An exact depth-$d$ circuit for MAJORITY yields
$\ket{\chi_{\mathrm{MAJ}_n}}$ at depth $2d$ by
\eqref{eq:phase-uncompute}. Applying \cref{cor:additive-variance-bound}
to this circuit gives $\sqrt n\le O(K^{2d}\cdot w^{6d+1})$.
Simplifying gives the width bound.
\end{proof}

\subsection{\texorpdfstring{Width Reduction}{Width Reduction}}
\label{sec:width-recursion}

\Cref{thm:width-recursion} converts a polynomial-size PARITY family into narrow two-dimensional circuits. We allow both the depth of the base family and the number of recursion levels to grow with the input length, and keep their dependence explicit.

The circuit computes parities of smaller blocks in parallel and combines their outputs in a new band of rows. Choosing more recursion levels decreases the width exponent at the cost of greater depth.

\begin{theorem}
\label{thm:width-recursion}
For a constant $c\ge1$, suppose a $\QAC$ family computes exact clean
$\Parity_n$ with depth $d(n)$ and $O(n^c)$ qubits and gates.
For every $n\ge2$ and integer $1\le r\le\lfloor\log_2 n\rfloor$,
there is a polynomial-size
$\left[O\left(rn^{c/r}\right),O\left(2^r n^{1+(c-1)/r}\right)\right]$-$\twoDQAC$
circuit that computes exact clean $\Parity_n$ with depth
$O\left(r\,d\left(\lceil n^{1/r}\rceil\right)\right)$.
The implicit constants are independent of $n$ and $r$.
\end{theorem}

\begin{remark}[Relaxing the clean hypothesis]
Exact decision families can first be cleaned. Approximate clean families
with sufficiently high constant fidelity can first be simulated and
exactified using
\cref{thm:two-dimensional-simulation,thm:exact-amplification}.
These conversions preserve
polynomial resources, so the recursion applies to these families as well.
\end{remark}

\begin{proof}[Proof of \cref{thm:width-recursion}]
We first give a circuit for $N=m^r$ inputs, where $m\ge2$.
By \cref{thm:two-dimensional-simulation}, the $m$-input base circuit
has a clean geometric simulation using $O(m^c)$ rows and columns and
depth 
$O(d(m))$. Place all logical qubits on the top data row, with the
$m$ inputs first and the accumulator last. The simulation workspace
lies below this row and starts and ends in $\ket1$.

For one level, use this simulator with zero accumulator. At each
later level, place $m$ child circuits side by side, with their outputs
on the common top boundary in distinct columns. Add a simulator band
of $O(m^c)$ rows above them, with its data row at the top and its zero
input qubits aligned with the child outputs. An additional $O(m^c)$
columns to the right hold the remaining logical qubits, including the
accumulator at the rightmost edge. Fill the gaps between logical
columns in the new band with qubits initialized to $\ket1$.

The simulator accommodates these gaps by extending each horizontal
gate across their fixed $\ket1$ qubits. Vertical swaps remain within
the original logical columns, so this changes only the length.

Before running the parent simulator, copy the child outputs to its
inputs using vertical interval CNOTs. Each interval crosses only
unused workspace in the new band, still in $\ket1$, and the distinct
columns allow all copies to run in parallel. The parent then computes
the parity of these bits on its top boundary. Retain the child
computations and copied inputs for uncomputation at the end.

Let $W_j$, $h_j$, and $d_j$ denote the width, length, and depth after
$j$ levels. Starting from $W_1,h_1=O(m^c)$ and 
$d_1=O(d(m))$, we have
\begin{itemize}
\item $W_j\le W_{j-1}+O(m^c)$;
\item $h_j\le mh_{j-1}+O(m^c)$;
\item 
$d_j\le d_{j-1}+O(d(m))$.
\end{itemize}
The implicit constants depend only on $c$ and the base family.
Solving these recurrences and substituting $m=N^{1/r}$ gives
$W_r=O\left(rN^{c/r}\right)$, $h_r=O\left(N^{1+(c-1)/r}\right)$,
and 
$d_r=O(rd(m))$.

To obtain clean computation, copy the final boundary output and
reverse the forward circuit. This preserves the asymptotic resource bounds.
For arbitrary $n\ge2$ and $1\le r\le\lfloor\log_2 n\rfloor$,
choose $m=\lceil n^{1/r}\rceil\ge2$ and pad with zeros to $N=m^r$.
Then
\[
 n\le N\le2^r n\le n^2.
\]
Substituting $m\le2n^{1/r}$ gives the claimed width and length.
The grid has $O(rNm^{2c-1})$ qubits, which is polynomial in $n$
throughout this range of $r$. Since the base circuits have polynomial
size, the resulting depth and gate count are polynomial as well.
\end{proof}

\subsection{Lower Bounds for \texorpdfstring{$\oneDQAC$}{1D-QAC} Functions}
\label{sec:one-dimensional-depth-lower-bounds}

We now ask how deep a one-dimensional circuit must be to compute
Boolean functions of many input bits. An interval gate can span many
qubits in one step, so its reach alone does not give a depth bound.
Throughout this subsection, the $n$ classical input qubits may occupy
arbitrary positions on the line, with any finite number of ancillas
in an input-independent state, possibly mixed or entangled. We measure
one designated output qubit and write $g_C(x)$ for the resulting bit.
\Cref{thm:one-dimensional-parity-lower-bound} gives the bound for PARITY, while
\cref{thm:one-dimensional-majority-mod-lower-bound} gives the bound for MAJORITY
and MOD.

\begin{theorem}
\label{thm:one-dimensional-parity-lower-bound}
For uniformly random $x\in\bits^n$, every depth-$d$ $\oneDQAC$
circuit $C$ in this setting satisfies
\[
 \Pr_{x,C}[g_C(x)=\Parity_n(x)]\le\frac12+81^d\cdot 2^{-n}.
\]
In particular, approximate computation of PARITY, with success
probability at least $2/3$ on every input, requires depth $\Omega(n)$.
\end{theorem}

\begin{proof}
Expanding the output projector backward gives
\[
 C^\dagger\proj{1}_{\mathrm{out}}C
 =\sum_\alpha T_\alpha\otimes I_{J_\alpha^c},
\]
where each $J_\alpha$ is an interval containing the output. Let
$R=I-2\Pi\otimes\Pi_{\mathrm{out}}$ cross an endpoint of $J_\alpha$,
where $\Pi$ acts on $J_\alpha$  and
$\Pi_{\mathrm{out}}$ is rank one on the newly reached qubits. Then
\[
 \begin{aligned}
 R(T_\alpha\otimes I)R
 &=T_\alpha\otimes I-2\Pi T_\alpha\otimes\Pi_{\mathrm{out}}
   -2T_\alpha\Pi\otimes\Pi_{\mathrm{out}}
   +4\Pi T_\alpha\Pi\otimes\Pi_{\mathrm{out}}.
 \end{aligned}
\]
The first term retains $J_\alpha$. Since
$\norm{\Pi}_\infty=\norm{\Pi_{\mathrm{out}}}_1=1$, the sum of coefficient
trace norms after the crossing is at most
$
 \norm{T_\alpha}_1+2\norm{\Pi T_\alpha}_1
 +2\norm{T_\alpha\Pi}_1+4\norm{\Pi T_\alpha\Pi}_1
 \le9\norm{T_\alpha}_1.
$
At most two gates cross the endpoints in a layer; all other gates
preserve trace norm. Thus
\[
 \sum_\alpha\norm{T_\alpha}_1
 \le9^{2d}\norm{\proj{1}}_1=81^d.
\]

Let $p(x)=\Pr_C[g_C(x)=1]$. Its full Fourier coefficient satisfies
\[
\abs{ \Pr_{x,C}[g_C(x)=\Parity_n(x)] - \frac{1}{2}} = \abs{\widehat p([n])}
 \le2^{-n}\sum_\alpha\norm{T_\alpha}_1
 \le81^d\cdot 2^{-n},
\]
which forces $d=\Omega(n)$.
\end{proof}

\begin{theorem}
\label{thm:one-dimensional-majority-mod-lower-bound}
In the same setting, MAJORITY and fixed-modulus MOD require depth
$\Omega(n)$ for approximate computation.
\end{theorem}

\begin{proof}
Let $C$ be a depth-$d$ circuit that approximately computes $f$, where
$f$ is either function, and write $p(x)=\Pr_C[g_C(x)=1]$. Let $I$ be the
input positions. For $s\ge1$, let $S\subseteq[n]$ index the nearest $s$
inputs at or to the left of the output and the nearest $s$ to its right.
Then $|S|\le2s$; write $S^c=[n]\setminus S$ and set
\[
 p_S(x_S)=\E_{x_{S^c}\sim\bits^{S^c}}\,p(x_S,x_{S^c}).
\]

Now,
\[
 \begin{aligned}
 &\E_{x_S\sim\bits^S}
 \left|\Pr_{x_{S^c}\sim\bits^{S^c}}
 [f(x_S,x_{S^c})=b]-\frac12\right|\\
 &\quad=\frac12-\E_{x_S\sim\bits^S}
 \min_{c\in\bits}\Pr_{x_{S^c}\sim\bits^{S^c}}
 [f(x_S,x_{S^c})=c]\\
 &\quad\ge\frac12-3 \cdot \E_{x\sim\bits^n}|p(x)-p_S(x_S)|\\
 &\quad\ge\frac12-6 \cdot \sum_{\alpha:\,|J_\alpha\cap I|>s}
 2^{-|J_\alpha\cap I|} 
\cdot \norm{T_\alpha}_1\\
 &\quad\ge\frac12-3\cdot81^d\cdot2^{-s}.
 \end{aligned}
\]
where we use the fact $J_\alpha\cap I\subseteq S$ whenever $|J_\alpha\cap I|\le s$.

For MAJORITY, take $s$ to be a sufficiently small fixed fraction of
$n$. Conditional on $x_S$, the remaining Hamming weight is binomial
on at least $n/2$ bits, with largest point mass $O(n^{-1/2})$. Since
$\E_{x_S\sim\bits^S}\bigl||x_S|-|S|/2\bigr|=O(\sqrt{s})$, we have
\[
 \E_{x_S\sim\bits^S}
 \left|\Pr_{x_{S^c}\sim\bits^{S^c}}
 [\mathrm{MAJ}_n(x_S,x_{S^c})=b]-\frac12\right|
 =O\!\left(\sqrt{\frac{s}{n}}\right).
\]
For a small enough fixed fraction and sufficiently large $n$, this
expectation is at most $1/4$. The preceding relation then implies
$d=\Omega(n)$.

For MOD modulo fixed $r\ge2$, take $s=\floor{n/4}$. The remaining at
least $n/2$ input bits are uniform, so the roots-of-unity filter in
\cref{thm:mod-sensitivity} gives, uniformly in $x_S$ and $b$,
\[
 \left|\Pr_{x_{S^c}\sim\bits^{S^c}}
 [\mathrm{MOD}_{r,n}(x_S,x_{S^c})=b]-\frac12\right|
 =\frac12-\frac1r+o(1).
\]
The gap from $1/2$ is $1/r+o(1)$, so
$d=\Omega(n)$.
\end{proof}

\section{Synthesis in \texorpdfstring{$\twoDQAC$}{2D-QAC}}
\label{sec:parity-constructions}

We give exact clean synthesis bounds for Boolean unitaries and
cat states. \Cref{thm:exact-parity-construction,thm:boolean-synthesis}
give width--depth tradeoffs for parity, MAJORITY, and fixed-modulus MOD.
\Cref{thm:one-dimensional-boolean-synthesis} gives linear-depth
one-dimensional constructions for MAJORITY and fixed-modulus MOD.
\Cref{thm:one-dimensional-cat-construction} gives the corresponding
cat-state synthesis bounds, together with an ancilla-free algorithm
in one dimension.

\subsection{Boolean functions}
\label{sec:exact-parity-construction}
\label{sec:bounds-comparison}

We first give exact clean parity constructions at several depths
and widths in \cref{thm:exact-parity-construction}. The MAJORITY and
fixed-modulus MOD constructions appear in
\cref{thm:boolean-synthesis} for narrow two-dimensional layouts and
in \cref{thm:one-dimensional-boolean-synthesis} for one dimension.

\begin{theorem}
\label{thm:exact-parity-construction}
For $n\ge2$, the parity unitary on $n$ controls has the following exact
clean implementations:
\begin{enumerate}
\renewcommand{\labelenumi}{(\alph{enumi})}
\item A depth-$O\left(n\right)$ $\left[1,n+1\right]$-$\twoDQAC$ circuit with no ancillas.
\item A depth-$O\left(\log n\right)$ $\left[2,n+1\right]$-$\twoDQAC$ circuit.
\item A depth-$O\left(\log n/\log\log n\right)$
$\left[\poly(n),\poly(n)\right]$-$\twoDQAC$ circuit.
\item For every integer $1\le d\le\log_2 n$, a depth-$O\left(d\right)$
$\left[O\left(dn^{1/d}\right),\exp\left(O\left(n^{1/d}\log n\right)\right)\right]$-$\twoDQAC$ circuit.
\end{enumerate}
\end{theorem}

We use the following synthesis result of Rosenthal.

\begin{lemma}[{\cite[Theorem~4.1 and Section~4.2]{rosenthal:LIPIcs.ITCS.2021.32}}]
\label{lem:rosenthal-preparation}
For every $m\ge2$, there is a depth-2
$\left[m,\exp(O\left(m\log m\right))\right]$-$\twoDQAC$ circuit synthesizing a
$3/4$-approximate $m$-nekomata.
\end{lemma}

We combine this synthesis with the amplification and
geometric conversions from
\cref{lem:endpoint-amplification,prop:parity-fanout-reductions} to prove \cref{thm:exact-parity-construction}.

\begin{proof}[Proof of \cref{thm:exact-parity-construction}]
\textbf{(a)} Arrange the $n$ controls and the accumulator consecutively in the
single row of a $\oneDQAC$ circuit. Apply CNOTs from control $j$ to control $j+1$ for
$j=1,\ldots,n-1$, then from the last control to the accumulator.
Reverse the first $n-1$ gates to restore the controls. All gates act
on adjacent qubits, giving depth $O\left(n\right)$ without ancillas.

\par\medskip\noindent\textbf{(b)} Arrange the controls and accumulator in the first row and
initialize the second row to $\ket1$. Pair consecutive controls,
compute each pair's parity into its right control, and repeat on the
surviving controls. An unpaired control passes to the next round.
After $\ceil{\log_2 n}$ rounds, the rightmost control holds the parity.
Copy it to the adjacent accumulator and reverse the aggregation.

To implement a round geometrically, swap the two controls of each
pair into the second row, apply a CNOT along that row, and swap them
back. The intermediate qubits in each interval are all in $\ket1$,
so a generalized Toffoli implements the CNOT. The intervals for
different pairs are disjoint, and the second row is restored after
each round. Thus, each round takes constant depth, giving the claimed
$[2,n+1]$ circuit. Single-qubit $X$ gates convert the ancilla
initialization and restoration to the all-zero convention.

\par\medskip\noindent\textbf{(c)} By \cref{lem:rosenthal-preparation,lem:endpoint-amplification,prop:parity-fanout-reductions},
parity on $m$ controls has a constant-depth clean $\QAC$ circuit
using $\exp(O\left(m\log m\right))$ qubits and gates. For sufficiently large $n$,
choose $m=\lfloor\log n/\log\log n\rfloor$ and arrange these circuits
in a balanced $m$-ary tree, padding with zeros as needed. The number
of levels and the resources per block satisfy
\[
 \begin{aligned}
 &\ceil{\log_m n}=O\left(\log n/\log\log n\right),\\
 &\exp(O\left(m\log m\right))=\poly(n).
 \end{aligned}
\]
There are polynomially many blocks. Compute their outputs from the
leaves upward, copy the final output into the accumulator, and reverse
the computation. This gives a clean $\QAC$ circuit of the required
depth and polynomial size. Applying
\cref{thm:two-dimensional-simulation} gives a
$\left[\poly(n),\poly(n)\right]$-$\twoDQAC$ circuit with the same asymptotic depth.

\par\medskip\noindent\textbf{(d)}
We first give a constant-depth $\left[O\left(m\right),\ell_m\right]$-$\twoDQAC$
clean parity gadget on $m$ controls, where $\ell_m=\exp(O\left(m\log m\right))$.
The synthesis circuit from \cref{lem:rosenthal-preparation} has its targets
in the left boundary column, so we can apply
\cref{lem:endpoint-amplification} with $p=3/4$.
Adding a constant number of target rows gives an exact nekomata
in constant geometric depth.
Keep any $m$ consecutive targets and regard the others as garbage.
By \cref{prop:parity-fanout-reductions}, this gives the desired clean
parity gadget.  Its inputs form a left boundary column, its
workspace is to the right of that column, and its accumulator can be
placed in an additional top row at the right boundary.  Enlarging
the constants accommodates the finitely many small values of $m$.

Set $m=\ceil{n^{1/d}}$ and $N=m^{d}$, padding the original input by
zero bits.  We define a forward computation for $m^j$ inputs by
induction on $j$.  Every child has its output on the top boundary of
its grid.  For $j=1$ use one gadget.  At the next level, put $m$
children side by side, and add a fresh band of $O\left(m\right)$ rows above them.
Place the parent gadget in this band, entirely to the right of the
horizontal span of the children.  Its input column is its leftmost
column.

The child outputs are in distinct columns.  Copy output $i$ vertically
to the row of parent input $i$, using a temporary qubit above that child.
Each vertical interval begins at the child's top boundary and otherwise
lies in the fresh band.  Its target is zero and its intermediate qubits
are one.  These copies have disjoint column supports.  Next, in each
input row copy its landing bit horizontally to the corresponding
parent input.  The horizontal interval has only one landing bit in
that row; every other intermediate qubit is a fresh one.  The parent
workspace lies strictly beyond the input column and is not traversed.
The horizontal copies have disjoint row supports.  Thus, a constant
number of legal interval layers routes all child outputs, including
on coherent superpositions, without crossing their working registers
or any active parent workspace.  The copied bits and temporary qubits
are retained as garbage in the forward computation.

Run the parent gadget and continue.  Different groups at the same
level use disjoint rectangles and run in parallel.  If $W_j,h_j$ and
$d_j$ denote the forward resources, the recursion gives
\begin{itemize}
\item $W_j\le W_{j-1}+O\left(m\right)$;
\item $h_j\le m\cdot h_{j-1}+O\left(\ell_m\right)$;
\item $d_j\le d_{j-1}+O\left(1\right)$.
\end{itemize}
The constants are independent of $j,m$.  With one gadget at $j=1$,
the width after $d$ levels is $W_{d}=O\left(dm\right)$ and the depth is
$O\left(d\right)$.  The length is $h_{d}=O\left(m^{d-1}\ell_m\right)$.
Multiplying width and length bounds the total physical qubits, and
hence the ancillas.  Finally, $m\le2n^{1/d}$ and
$N=m^{d}\le2^{d}n\le n^2$ in the stated range of $d$.
These estimates give the claimed width bound; the polynomial factors
in the length are absorbed into $\exp(O(n^{1/d}\log n))$.

Every block computes parity exactly before being used in the
recursion.  Hence, there is no error to accumulate over the blocks.
The final output lies on the outer top boundary.  Copy it into an
adjacent new accumulator and reverse the full forward computation,
as in \cref{prop:parity-fanout-reductions}.  This restores every
intermediate parity, temporary qubit, and workspace qubit, while
preserving the claimed asymptotic resources.  Local gates return the
fixed one-filled qubits to the all-zero ancilla convention.

\end{proof}

The parity gadget also lets us simulate circuits with fanout while
keeping their large auxiliary registers along the length of the grid.
The proof is given in Appendix~\ref{sec:cat-fanout-exactification}.

\begin{restatable}{lemma}{geometricfanoutsimulation}
\label{lem:geometric-fanout-simulation}
Every depth-$d\ge1$ circuit on $n+a$ qubits over $\QAC$ gates and
fanout gates with at most $r$ targets has an exact depth-$O(d)$
simulation by an
$[O(n+a),n+a+\exp(O(r\log(r+2)))]$-$\twoDQAC$ circuit.
The original qubits lie on the top boundary, and every
added ancilla is restored to its initial product state.
\end{restatable}

For MAJORITY and MOD, we combine local counts rather than local
function values. Exact arithmetic with fanout gives constant-depth
circuits for each merge; \cref{lem:geometric-fanout-simulation}
places these circuits in narrow bands.

\begin{theorem}
\label{thm:boolean-synthesis}
For every fixed $d\ge1$, the MAJORITY and fixed-modulus MOD unitaries
admit exact clean depth-$O(d)$ $\twoDQAC$ synthesis with unrestricted
product ancillas and width $n^{O(1/d)}$.
\end{theorem}

\begin{remark}
Together with \cref{thm:majority-sensitivity,thm:mod-sensitivity,thm:parity-sensitivity},
this gives optimal width $n^{\Theta(1/d)}$ for exact MAJORITY and
fixed-modulus MOD at sufficiently large fixed depths, with
unrestricted product ancillas.
\end{remark}

\begin{proof}[Proof of \cref{thm:boolean-synthesis}]
Arrange the inputs at the leaves of a height-$d$ tree with
$\ceil{n^{1/d}}$ children per node, padding with zeros as needed.
At each node, add the child counts into a fresh register. Every
count is at most $n$ and needs only $O(\log n)$ bits.
Iterated addition, comparison, and fixed-modulus reduction have
constant-depth, polynomial-size threshold circuits
\cite{hesse2002division}. Their exact quantum simulation with
fanout \cite{6597759}, followed by compute--copy--uncompute and
\cref{lem:geometric-fanout-simulation}, therefore implements each
merge in constant geometric depth and width
$(n^{1/d}\log n)^{O(1)}$.

Keep each child count on the top boundary of its rectangle. Place
the children side by side and the parent circuit in a fresh band
above them, entirely to their right. Its inputs and outputs lie on
its top row, with the inputs to the left of the outputs. Route each
child bit vertically to its own row above the parent, horizontally
to its input column, and then down to that input. Each step uses
disjoint intervals with fixed intermediate qubits. The input routes
avoid the output columns, so the new count can be copied vertically
to the top boundary. These routes add only $O(n^{1/d}\log n)$ rows
and constant depth. Retain all intermediate counts and copies until
the final reversal.

At the root, compare the count with $n/2$ for MAJORITY, or test it
modulo the fixed modulus for MOD. Copy the answer into the external
accumulator and reverse the entire computation, restoring all
ancillas. Nodes at each level occupy disjoint rectangles and run in
parallel. Each level adds $(n^{1/d}\log n)^{O(1)}$ rows and constant
depth, and the final test has no greater cost. Thus, for fixed $d$,
the total depth is $O(d)$ and the width is
$d(n^{1/d}\log n)^{O(1)}=n^{O(1/d)}$,
with an exponent constant independent of $d$.
\end{proof}

\begin{theorem}
\label{thm:one-dimensional-boolean-synthesis}
The MAJORITY and fixed-modulus MOD unitaries admit exact clean
depth-$O(n)$ $\oneDQAC$ synthesis with $O(\log n)$ and $O(1)$
ancillas, respectively.
\end{theorem}

\begin{proof}
For MAJORITY, let $k=\ceil{\log_2(n+1)}$. Place the accumulator, the
inputs, and a $k$-qubit counter in that order. Prepare the counter in
$\ket{+}^{\otimes k}$, with bits from most to least significant. Move
it left through the inputs by swapping adjacent input--counter pairs
in parallel. Before each swap, apply a controlled phase whose angle
is $2\pi/2^k$ times the counter bit's binary weight. After $n+k-1$
rounds every pair has crossed once, leaving the Fourier encoding of
$|x|$ next to the accumulator. Since $|x|<2^k$, the inverse Fourier
transform gives the binary count.

For each zero bit of $\floor{n/2}$, flip the accumulator when all more
significant count bits match the threshold and this bit is one. Exactly
one test fires if $|x|>n/2$. Each test is an interval generalized
Toffoli, with $X$ gates for negative controls, and all tests take
$O(k)$ depth. Reverse the transform, scan, and initialization to clean
the counter and restore the inputs. The transform takes $O(k^2)$ depth
on adjacent qubits, giving $O(n+k^2)=O(n)$ overall.

For MOD modulo fixed $r$, use a $\ceil{\log_2r}$-qubit counter
initialized to zero. Move it left through the inputs, adding each
input bit modulo $r$ before swapping past it. Extend the increment by
the identity on unused encodings. This fixed-size permutation is a
product of swaps of basis strings differing in one bit; each is an
interval generalized Toffoli with possible $X$ gates on its controls.
Test the final residue for zero, flip the accumulator, and undo the
scan. Each step has constant depth, giving depth $O(n)$.
\end{proof}

\subsection{States}
\label{sec:cat-state-construction}

We combine the parity circuits in
\cref{sec:exact-parity-construction} with
\cref{prop:parity-fanout-reductions} to synthesize cat states.

\begin{theorem}
\label{thm:one-dimensional-cat-construction}
For $n\ge2$, the $n$-qubit cat state admits exact synthesis by the
following circuits:
\begin{enumerate}
\renewcommand{\labelenumi}{(\alph{enumi})}
\item A depth-$\ceil{\log_2 n}$ $[1,n]$-$\twoDQAC$ circuit with no ancillas.
\item A depth-$O\left(\log n/\log\log n\right)$
$\left[\poly(n),\poly(n)\right]$-$\twoDQAC$ circuit.
\item For every integer $1\le d\le\log_2 n$, a depth-$O\left(d\right)$
$\left[O\left(dn^{1/d}\right),\exp\left(O\left(n^{1/d}\log n\right)\right)\right]$-$\twoDQAC$ circuit.
\end{enumerate}
\end{theorem}

\begin{remark}
For fixed depths covered by these circuits, the optimal width
for exact clean parity and cat-state synthesis is $n^{\Theta(1/d)}$
when ancillas are unrestricted. This follows by combining the
circuits with \cref{thm:width-lower-bound,prop:parity-fanout-reductions}.
The inverse-depth exponent is therefore unconditionally tight,
with exponential length and workspace allowed in the upper bounds.
\end{remark}

\begin{figure}[htp!]  
    \centering
    
\begin{quantikz}[row sep=0.3cm, column sep=0.4cm]
\lstick{$\ket{0}$} & \gate{H} &\qw & \ctrl{1}  & \qw &\ctrl{1}    & \qw &\ctrl{1}         & \qw & \rstick[8]{$\ket{\mathrm{cat}_8}$}\qw \\
\lstick{$\ket{0}$} & \gate{X} &\qw & \ctrl{1}  & \qw &\ctrl{1}    &\gate{X} &\targ{}  & \qw & \qw \\
\lstick{$\ket{0}$} & \gate{X} &\qw  & \ctrl{1}  &\gate{X}   &\targ{}     & \qw &\ctrl{1}         & \qw & \qw \\
\lstick{$\ket{0}$} & \gate{X} &\qw & \ctrl{1}  & \qw & \qw &\gate{X} &\targ{}  & \qw & \qw \\
\lstick{$\ket{0}$} & \gate{X} & \gate{X}     & \targ{}   & \qw &\ctrl{1}    & \qw & \ctrl{1}        & \qw & \qw \\
\lstick{$\ket{0}$} & \gate{X} &\qw & \qw       & \qw &\ctrl{1}    &\gate{X}& \targ{}  & \qw & \qw \\
\lstick{$\ket{0}$} & \gate{X}&\qw  & \qw       & \gate{X}  &\targ{}     & \qw &\ctrl{1}         & \qw & \qw \\
\lstick{$\ket{0}$} & \gate{X} &\qw & \qw       & \qw & \qw &\gate{X} &\targ{}  & \qw & \qw 
\end{quantikz}

    \caption{A $[1,8]$-$\twoDQAC$ circuit for generating the $8$-qubit cat state}
    \label{fig:one-dimensional-cat-construction}
\end{figure}

\begin{proof}[Proof of \cref{thm:one-dimensional-cat-construction}]
\textbf{(a)} We synthesize a $\oneDQAC$ unitary $F_n$ that maps
$\ket{b,1^{n-1}}$ to $\ket{b^n}$ for each $b\in\bits$.
For $n=1$, take $F_1=I$. For $n\ge2$, split the qubit register into a left
interval of length $\ceil{n/2}$ and a right interval of length
$\floor{n/2}$. Apply $X$ to the first qubit of the right interval,
then a generalized Toffoli controlled on every qubit of the left
interval and targeting that qubit. All controls except the first
are in $\ket1$, so this copies $b$ to the first qubit of the right
interval. The resulting state is
\[
 \ket{b,1^{\ceil{n/2}-1}}\otimes
 \ket{b,1^{\floor{n/2}-1}}.
\]
Apply $F_{\ceil{n/2}}$ and $F_{\floor{n/2}}$ recursively on the two
intervals in parallel. Every gate has contiguous support, and the
two recursive circuits have disjoint supports. The depth therefore
satisfies $D(1)=0$ and
$D(n)\le1+\max\{D(\ceil{n/2}),D(\floor{n/2})\}$, giving
$D(n)\le\ceil{\log_2 n}$.

Starting from $\ket{0^n}$, apply $H$ to the first qubit and $X$ to
all remaining qubits. By linearity, $F_n$ maps the resulting state
$\ket{+,1^{n-1}}$ to $\ket{\mathrm{cat}_n}$.
These single-qubit gates do not increase the number of multi-qubit
layers, and every qubit belongs to the output register.
\Cref{fig:one-dimensional-cat-construction} illustrates the circuit
for $n=8$.

\par\medskip\noindent\textbf{(b), (c)} These follow from parts (c) and (d)
of \cref{thm:exact-parity-construction}, respectively, via
\cref{prop:parity-fanout-reductions}, with the same asymptotic resource bounds.
\end{proof}

\subsection*{AI methodology}

We used OpenAI GPT 6 Astra to help check the correctness of proofs, refine the writing and proof exposition, improve the figures, and formalize the descriptions of algorithms and proofs.

This paper builds on an earlier version whose proof ideas were developed without AI assistance. In preparing the current version, we used AI to help develop new lower bounds for $\twoDQAC$ circuits and improve the lower bounds for $\oneDQAC$ circuits in Section~\ref{sec:lower-bounds}. More details can be
found in Sections~\ref{sec:width-lower-bound}, \ref{sec:fourier-concentration}, \ref{sec:variance-applications}, and~\ref{sec:one-dimensional-depth-lower-bounds}. For the 2D-QAC lower bounds, the authors provided an initial proof, and AI helped simplify it and reduce the exponent of the width parameter in the per-layer variance estimate from \(5\) to \(3\). For the 1D-QAC lower bounds, AI helped improve the depth bound from \(\Omega(n/\log n)\) in the earlier version to \(\Omega(n)\).

AI also helped rewrite the TikZ code for the figures illustrating the two-dimensional setting to make them easier to read. The authors verified the correctness and originality of all content, including references. The authors take full responsibility for the manuscript.

\bibliographystyle{alpha}
\bibliography{references}

@inproceedings{10.1145/3313276.3316404,
  author        = {Watts, Adam Bene and Kothari, Robin and Schaeffer, Luke and Tal, Avishay},
  title         = {Exponential separation between shallow quantum circuits and unbounded fan-in shallow classical circuits},
  booktitle     = {Proceedings of the 51st Annual ACM SIGACT Symposium on Theory of Computing},
  year          = {2019},
  pages         = {515--526},
  publisher     = {Association for Computing Machinery},
  doi           = {10.1145/3313276.3316404},
  url           = {https://doi.org/10.1145/3313276.3316404},
}

@inproceedings{10.1145/3357713.3384332,
  author        = {Grier, Daniel and Schaeffer, Luke},
  title         = {Interactive shallow {Clifford} circuits: Quantum advantage against {NC$^1$} and beyond},
  booktitle     = {Proceedings of the 52nd Annual ACM SIGACT Symposium on Theory of Computing},
  year          = {2020},
  pages         = {875--888},
  publisher     = {Association for Computing Machinery},
  doi           = {10.1145/3357713.3384332},
  url           = {https://doi.org/10.1145/3357713.3384332},
}

@inproceedings{10.1145/3618260.3649722,
  author        = {Huang, Hsin-Yuan and Liu, Yunchao and Broughton, Michael and Kim, Isaac and Anshu, Anurag and Landau, Zeph and McClean, Jarrod R.},
  title         = {Learning Shallow Quantum Circuits},
  booktitle     = {Proceedings of the 56th Annual ACM Symposium on Theory of Computing},
  year          = {2024},
  pages         = {1343--1351},
  publisher     = {Association for Computing Machinery},
  doi           = {10.1145/3618260.3649722},
  url           = {https://doi.org/10.1145/3618260.3649722},
}

@inproceedings{6597759,
  author        = {Takahashi, Yasuhiro and Tani, Seiichiro},
  title         = {Collapse of the Hierarchy of Constant-Depth Exact Quantum Circuits},
  booktitle     = {2013 IEEE Conference on Computational Complexity},
  year          = {2013},
  pages         = {168--178},
  doi           = {10.1109/CCC.2013.25},
  url           = {https://doi.org/10.1109/CCC.2013.25},
}

@inproceedings{9719811,
  author        = {Coble, Nolan J. and Coudron, Matthew},
  title         = {Quasi-polynomial Time Approximation of Output Probabilities of Geometrically-local, Shallow Quantum Circuits},
  booktitle     = {2021 IEEE 62nd Annual Symposium on Foundations of Computer Science (FOCS)},
  year          = {2022},
  pages         = {598--609},
  publisher     = {IEEE},
  doi           = {10.1109/FOCS52979.2021.00065},
  url           = {https://doi.org/10.1109/FOCS52979.2021.00065},
}

@inproceedings{ADOY24,
  author        = {Anshu, Anurag and Dong, Yangjing and Ou, Fengning and Yao, Penghui},
  title         = {On the Computational Power of {QAC}{$^0$} with Barely Superlinear Ancillae},
  booktitle     = {Proceedings of the 57th Annual ACM Symposium on Theory of Computing},
  year          = {2025},
  pages         = {1476--1487},
  publisher     = {Association for Computing Machinery},
  doi           = {10.1145/3717823.3718189},
  url           = {https://doi.org/10.1145/3717823.3718189},
}

@inproceedings{arunachalam2022learningphasestates,
  author        = {Arunachalam, Srinivasan and Bravyi, Sergey and Dutt, Arkopal and Yoder, Theodore J.},
  title         = {Optimal Algorithms for Learning Quantum Phase States},
  booktitle     = {18th Conference on the Theory of Quantum Computation, Communication and Cryptography (TQC 2023)},
  year          = {2023},
  series        = {Leibniz International Proceedings in Informatics (LIPIcs)},
  volume        = {266},
  pages         = {3:1--3:24},
  publisher     = {Schloss Dagstuhl -- Leibniz-Zentrum f{\"u}r Informatik},
  doi           = {10.4230/LIPIcs.TQC.2023.3},
  url           = {https://doi.org/10.4230/LIPIcs.TQC.2023.3},
  eprint        = {2208.07851},
  archivePrefix = {arXiv},
  primaryClass  = {quant-ph},
}

@inproceedings{foxman2025randomunitariesconstantquantum,
  author        = {Foxman, Ben and Parham, Natalie and Vasconcelos, Francisca and Yuen, Henry},
  title         = {Random Unitaries in Constant (Quantum) Time},
  booktitle     = {17th Innovations in Theoretical Computer Science Conference (ITCS 2026)},
  year          = {2026},
  editor        = {Saraf, Shubhangi},
  series        = {Leibniz International Proceedings in Informatics (LIPIcs)},
  volume        = {362},
  pages         = {61:1--61:25},
  publisher     = {Schloss Dagstuhl -- Leibniz-Zentrum f{\"u}r Informatik},
  doi           = {10.4230/LIPIcs.ITCS.2026.61},
  url           = {https://doi.org/10.4230/LIPIcs.ITCS.2026.61},
}

@inproceedings{grier2025quantumadvantagesamplingshallow,
  author        = {Grier, Daniel and Kane, Daniel M. and Morris, Jackson and Ostuni, Anthony and Wu, Kewen},
  title         = {Quantum Advantage from Sampling Shallow Circuits: Beyond Hardness of Marginals},
  booktitle     = {17th Innovations in Theoretical Computer Science Conference (ITCS 2026)},
  year          = {2026},
  editor        = {Saraf, Shubhangi},
  series        = {Leibniz International Proceedings in Informatics (LIPIcs)},
  volume        = {362},
  pages         = {73:1--73:14},
  publisher     = {Schloss Dagstuhl -- Leibniz-Zentrum f{\"u}r Informatik},
  doi           = {10.4230/LIPIcs.ITCS.2026.73},
  url           = {https://doi.org/10.4230/LIPIcs.ITCS.2026.73},
}

@inproceedings{Grier:2024xxt,
  author        = {Grier, Daniel and Morris, Jackson},
  title         = {Quantum Threshold Is Powerful},
  booktitle     = {40th Computational Complexity Conference (CCC 2025)},
  year          = {2025},
  series        = {Leibniz International Proceedings in Informatics (LIPIcs)},
  volume        = {339},
  pages         = {3:1--3:23},
  publisher     = {Schloss Dagstuhl -- Leibniz-Zentrum f{\"u}r Informatik},
  doi           = {10.4230/LIPIcs.CCC.2025.3},
  url           = {https://doi.org/10.4230/LIPIcs.CCC.2025.3},
}

@inproceedings{joshi2025improvedlowerboundsqac0,
  author        = {Joshi, Malvika Raj and Tal, Avishay and Vasconcelos, Francisca and Wright, John},
  title         = {Improved Lower Bounds for {QAC}{$^0$}},
  booktitle     = {Proceedings of the 58th Annual ACM Symposium on Theory of Computing},
  year          = {2026},
  pages         = {2199--2209},
  publisher     = {Association for Computing Machinery},
  doi           = {10.1145/3798129.3800922},
  url           = {https://doi.org/10.1145/3798129.3800922},
  eprint        = {2512.14643},
  archivePrefix = {arXiv},
}

@inproceedings{khavari2021pseudodimension,
  author        = {Khavari, Behnoush and Rabusseau, Guillaume},
  title         = {Lower and Upper Bounds on the Pseudo-Dimension of Tensor Network Models},
  booktitle     = {Advances in Neural Information Processing Systems},
  year          = {2021},
  volume        = {34},
  pages         = {10931--10943},
  publisher     = {Curran Associates, Inc.},
  url           = {https://proceedings.neurips.cc/paper_files/paper/2021/hash/5a9d8bf5b7a4b35f3110dde8673bdda2-Abstract.html},
}

@inproceedings{NPVY24,
  author        = {Nadimpalli, Shivam and Parham, Natalie and Vasconcelos, Francisca and Yuen, Henry},
  title         = {On the {Pauli} Spectrum of {QAC}{$^0$}},
  booktitle     = {Proceedings of the 56th Annual ACM Symposium on Theory of Computing},
  year          = {2024},
  pages         = {1498--1506},
  publisher     = {Association for Computing Machinery},
  doi           = {10.1145/3618260.3649662},
  url           = {https://doi.org/10.1145/3618260.3649662},
}

@inproceedings{rosenthal:LIPIcs.ITCS.2021.32,
  author        = {Rosenthal, Gregory},
  title         = {Bounds on the {QAC}{$^0$} Complexity of Approximating Parity},
  booktitle     = {12th Innovations in Theoretical Computer Science Conference (ITCS 2021)},
  year          = {2021},
  editor        = {Lee, James R.},
  series        = {Leibniz International Proceedings in Informatics (LIPIcs)},
  volume        = {185},
  pages         = {32:1--32:20},
  publisher     = {Schloss Dagstuhl -- Leibniz-Zentrum f{\"u}r Informatik},
  doi           = {10.4230/LIPIcs.ITCS.2021.32},
  url           = {https://doi.org/10.4230/LIPIcs.ITCS.2021.32},
}

@inproceedings{vasconcelos2024learning,
  author        = {Vasconcelos, Francisca and Huang, Hsin-Yuan},
  title         = {Learning shallow quantum circuits with many-qubit gates},
  booktitle     = {Proceedings of the 38th Annual Conference on Learning Theory (COLT)},
  year          = {2025},
  series        = {Proceedings of Machine Learning Research},
  volume        = {291},
  pages         = {5553--5604},
  publisher     = {PMLR},
  url           = {https://proceedings.mlr.press/v291/vasconcelos25a.html},
}

@inproceedings{watts2024unconditionalquantumadvantagesampling,
  author        = {Bene Watts, Adam and Parham, Natalie},
  title         = {Unconditional Quantum Advantage for Sampling with Shallow Circuits},
  booktitle     = {17th Innovations in Theoretical Computer Science Conference (ITCS 2026)},
  year          = {2026},
  editor        = {Saraf, Shubhangi},
  series        = {Leibniz International Proceedings in Informatics (LIPIcs)},
  volume        = {362},
  pages         = {17:1--17:12},
  publisher     = {Schloss Dagstuhl -- Leibniz-Zentrum f{\"u}r Informatik},
  doi           = {10.4230/LIPIcs.ITCS.2026.17},
  url           = {https://doi.org/10.4230/LIPIcs.ITCS.2026.17},
}

@article{10.1016/j.ipl.2011.05.002,
  author        = {Bera, Debajyoti},
  title         = {A lower bound method for quantum circuits},
  journal       = {Information Processing Letters},
  year          = {2011},
  volume        = {111},
  number        = {15},
  pages         = {723--726},
  doi           = {10.1016/j.ipl.2011.05.002},
  url           = {https://doi.org/10.1016/j.ipl.2011.05.002},
}

@article{10.5555/2011679.2011682,
  author        = {Fang, M. and Fenner, S. and Green, F. and Homer, S. and Zhang, Y.},
  title         = {Quantum lower bounds for fanout},
  journal       = {Quantum Information and Computation},
  year          = {2006},
  volume        = {6},
  number        = {1},
  pages         = {46--57},
}

@article{Chuang_1997,
  author        = {Chuang, Isaac L. and Nielsen, M. A.},
  title         = {Prescription for experimental determination of the dynamics of a quantum black box},
  journal       = {Journal of Modern Optics},
  year          = {1997},
  volume        = {44},
  number        = {11--12},
  pages         = {2455--2467},
  doi           = {10.1080/09500349708231894},
  url           = {https://doi.org/10.1080/09500349708231894},
}

@article{doi:10.1126/science.aar3106,
  author        = {Sergey Bravyi and David Gosset and Robert K{\"o}nig},
  title         = {Quantum advantage with shallow circuits},
  journal       = {Science},
  year          = {2018},
  volume        = {362},
  number        = {6412},
  pages         = {308--311},
  doi           = {10.1126/science.aar3106},
  url           = {https://doi.org/10.1126/science.aar3106},
}

@article{hanneke2016optimal,
  author        = {Hanneke, Steve},
  title         = {The Optimal Sample Complexity of {PAC} Learning},
  journal       = {Journal of Machine Learning Research},
  year          = {2016},
  volume        = {17},
  number        = {38},
  pages         = {1--15},
  url           = {https://jmlr.org/papers/v17/15-389.html},
}

@article{hesse2002division,
  author        = {Hesse, William and Allender, Eric and Barrington, David A. Mix},
  title         = {Uniform Constant-Depth Threshold Circuits for Division and Iterated Multiplication},
  journal       = {Journal of Computer and System Sciences},
  year          = {2002},
  volume        = {65},
  number        = {4},
  pages         = {695--716},
  url           = {https://people.cs.rutgers.edu/~allender/papers/division.pdf},
}

@article{hoyer2005quantum,
  author        = {H{\o}yer, Peter and {\v{S}}palek, Robert},
  title         = {Quantum fan-out is powerful},
  journal       = {Theory of Computing},
  year          = {2005},
  volume        = {1},
  number        = {5},
  pages         = {81--103},
  doi           = {10.4086/toc.2005.v001a005},
  url           = {https://theoryofcomputing.org/articles/v001a005/},
}

@article{moor,
  author        = {Green, Frederic and Homer, Steven and Moore, Cristopher and Pollett, Christopher},
  title         = {Counting, Fanout, and the Complexity of Quantum {ACC}},
  journal       = {Quantum Information and Computation},
  year          = {2002},
  volume        = {2},
  number        = {1},
  pages         = {35--65},
  doi           = {10.26421/QIC2.1-3},
  url           = {https://doi.org/10.26421/QIC2.1-3},
}

@article{vidal2003slightly,
  author        = {Vidal, Guifr{\'e}},
  title         = {Efficient Classical Simulation of Slightly Entangled Quantum Computations},
  journal       = {Physical Review Letters},
  year          = {2003},
  volume        = {91},
  number        = {14},
  pages         = {147902},
  doi           = {10.1103/PhysRevLett.91.147902},
  url           = {https://doi.org/10.1103/PhysRevLett.91.147902},
  eprint        = {quant-ph/0301063},
  archivePrefix = {arXiv},
}

@misc{dong2025linearsizeqac0channelslearning,
  author        = {Dong, Yangjing and Ou, Fengning and Yao, Penghui},
  title         = {Linear-Size {QAC}{$^0$} Channels: Learning, Testing and Hardness},
  year          = {2025},
  howpublished  = {arXiv preprint arXiv:2510.00593},
  url           = {https://arxiv.org/abs/2510.00593},
  eprint        = {2510.00593},
  archivePrefix = {arXiv},
  primaryClass  = {quant-ph},
}

@misc{fenner2025tightboundsdepth2qaccircuits,
  author        = {Fenner, Stephen and Grier, Daniel and Pad{\'e}, Daniel and Thierauf, Thomas},
  title         = {Tight bounds on depth-2 {QAC}-circuits computing parity},
  year          = {2025},
  howpublished  = {arXiv preprint arXiv:2504.06433},
  url           = {https://arxiv.org/abs/2504.06433},
  eprint        = {2504.06433},
  archivePrefix = {arXiv},
  primaryClass  = {quant-ph},
}

@misc{GGJ26,
  author        = {Gretta, Lucas and Gupta, Meghal and Joshi, Malvika Raj},
  title         = {Parity {$\notin$} {QAC}{$^0$} {$\iff$} {QAC}{$^0$} Is Fourier-Concentrated},
  year          = {2026},
  howpublished  = {arXiv preprint arXiv:2604.02793},
  url           = {https://arxiv.org/abs/2604.02793},
  eprint        = {2604.02793},
  archivePrefix = {arXiv},
  primaryClass  = {quant-ph},
  note          = {Accepted at {FOCS} 2026},
}

@misc{gheorghiu2026verifiable,
  author        = {Gheorghiu, Alexandru},
  title         = {Verifiable quantum advantage in extremely low depth},
  year          = {2026},
  howpublished  = {arXiv preprint arXiv:2609.01448},
  url           = {https://arxiv.org/abs/2609.01448},
  eprint        = {2609.01448},
  archivePrefix = {arXiv},
  primaryClass  = {quant-ph},
}

@misc{gretta2026logarithmicstate,
  author        = {Gretta, Lucas and Gupta, Meghal and Joshi, Malvika Raj},
  title         = {{QAC}{$^0$} Can Prepare Every Logarithmic-Qubit State},
  year          = {2026},
  howpublished  = {arXiv preprint arXiv:2609.17408},
  url           = {https://arxiv.org/abs/2609.17408},
  eprint        = {2609.17408},
  archivePrefix = {arXiv},
  primaryClass  = {quant-ph},
}

@misc{gretta2026polylogdicke,
  author        = {Gretta, Lucas and Gupta, Meghal and Joshi, Malvika Raj},
  title         = {Polylogarithmic-Weight {Dicke} States in {QAC}{$^0$} and Arbitrary Symmetric States in {QAC}{$^0_f$}},
  year          = {2026},
  howpublished  = {arXiv preprint arXiv:2604.15298},
  url           = {https://arxiv.org/abs/2604.15298},
  eprint        = {2604.15298},
  archivePrefix = {arXiv},
  primaryClass  = {quant-ph},
  note          = {To appear in {SODA} 2027},
}

@misc{grier2026mathsfqac0containsmathsftc0with,
  author        = {Grier, Daniel and Morris, Jackson and Wu, Kewen},
  title         = {{$\mathsf{QAC}^0$} Contains {$\mathsf{TC}^0$} (with Many Copies of the Input)},
  year          = {2026},
  howpublished  = {arXiv preprint arXiv:2601.03243},
  url           = {https://arxiv.org/abs/2601.03243},
  eprint        = {2601.03243},
  archivePrefix = {arXiv},
  primaryClass  = {cs.CC},
  note          = {Accepted at {FOCS} 2026},
}

@misc{grier2026robustness,
  author        = {Grier, Daniel and Morris, Jackson and Wu, Kewen},
  title         = {The Robustness of {$\mathsf{QAC}^0$}},
  year          = {2026},
  howpublished  = {arXiv preprint arXiv:2610.02154},
  url           = {https://arxiv.org/abs/2610.02154},
  eprint        = {2610.02154},
  archivePrefix = {arXiv},
  primaryClass  = {quant-ph},
}

@misc{moore1999quantum,
  author        = {Moore, Cristopher},
  title         = {Quantum circuits: Fanout, parity, and counting},
  year          = {1999},
  howpublished  = {arXiv preprint arXiv:quant-ph/9903046},
  url           = {https://arxiv.org/abs/quant-ph/9903046},
  eprint        = {quant-ph/9903046},
  archivePrefix = {arXiv},
  primaryClass  = {quant-ph},
}

@misc{openai2026localization,
  author        = {{OpenAI}},
  title         = {Product-projection localization and the {QAC}{$^0$} parity lower bound},
  year          = {2026},
  howpublished  = {{OpenAI} preprint},
  url           = {https://github.com/openai/math/blob/main/preprints/Product-projection-localization-and-the-QAC0-parity-lower-bound-September-24-2026/paper.pdf},
  note          = {September 24, 2026},
}

@misc{openai2026trajectories,
  author        = {{OpenAI}},
  title         = {Regular trajectories, pruning and quantum parity},
  year          = {2026},
  howpublished  = {{OpenAI} preprint},
  url           = {https://github.com/openai/math/blob/main/preprints/Regular-trajectories-pruning-and-quantum-parity-September-24-2026/paper.pdf},
  note          = {September 24, 2026},
}

@misc{vasconcelos2026constantdepthunitarypreparationdicke,
  author        = {Joshi, Malvika Raj and Vasconcelos, Francisca},
  title         = {Constant-Depth Unitary Preparation of {Dicke} States},
  year          = {2026},
  howpublished  = {arXiv preprint arXiv:2601.10693},
  url           = {https://arxiv.org/abs/2601.10693},
  eprint        = {2601.10693},
  archivePrefix = {arXiv},
  primaryClass  = {quant-ph},
}

@misc{vasconcelos2026forrelation,
  author        = {Vasconcelos, Francisca},
  title         = {2-Fold Forrelation is in {QAC}{$^0$}},
  year          = {2026},
  howpublished  = {arXiv preprint arXiv:2609.07060},
  url           = {https://arxiv.org/abs/2609.07060},
  eprint        = {2609.07060},
  archivePrefix = {arXiv},
  primaryClass  = {quant-ph},
}

@techreport{gao2026simplifiedparity,
  author        = {Gao, Minbo and Ji, Zhengfeng and Liu, Chenghua},
  title         = {A simplified proof of {PARITY} {$\notin$} {QAC}{$^0$}},
  year          = {2026},
  number        = {TR26-258},
  institution   = {Electronic Colloquium on Computational Complexity},
  url           = {https://eccc.weizmann.ac.il/report/2026/258/},
}

@misc{WillowSpec,
  author        = {{Google Quantum AI}},
  title         = {{Willow} Spec Sheet},
  year          = {2024},
  howpublished  = {Technical specification},
  url           = {https://quantumai.google/static/site-assets/downloads/willow-spec-sheet.pdf},
  note          = {Accessed: 2026-10-28},
}

@article{kim2023utility,
  author = {Kim, Youngseok and Eddins, Andrew and Anand, Sajant
            and Wei, Ken Xuan and van den Berg, Ewout
            and Rosenblatt, Sami and Nayfeh, Hasan and Wu, Yantao
            and Zaletel, Michael and Temme, Kristan
            and Kandala, Abhinav},
  title = {Evidence for the utility of quantum computing
           before fault tolerance},
  journal = {Nature},
  volume = {618},
  pages = {500--505},
  year = {2023},
  doi = {10.1038/s41586-023-06096-3},
  url = {https://doi.org/10.1038/s41586-023-06096-3}
}

@article{hendrickx2021four,
  author = {Hendrickx, Nico W. and Lawrie, William I. L.
            and Russ, Maximilian and van Riggelen, Floor
            and de Snoo, Sander L. and Schouten, Raymond N.
            and Sammak, Amir and Scappucci, Giordano
            and Veldhorst, Menno},
  title = {A four-qubit germanium quantum processor},
  journal = {Nature},
  volume = {591},
  pages = {580--585},
  year = {2021},
  doi = {10.1038/s41586-021-03332-6},
  url = {https://doi.org/10.1038/s41586-021-03332-6}
}

@misc{strelchuk2026power,
  author        = {Strelchuk, Sergii and Subramanian, Sathyawageeswar
                   and Weisz, M{\'a}t{\'e}},
  title         = {The power of constant-depth quantum circuits of unbounded size},
  year          = {2026},
  eprint        = {2609.40294},
  archivePrefix = {arXiv},
  primaryClass  = {quant-ph},
  doi           = {10.48550/arXiv.2609.40294},
  url           = {https://arxiv.org/abs/2609.40294}
}

@misc{nehoran2026all,
  author        = {Nehoran, Barak and Slote, Joseph and Yuen, Henry},
  title         = {All Unitaries Have Constant Depth Quantum Circuits},
  year          = {2026},
  eprint        = {2609.40351},
  archivePrefix = {arXiv},
  primaryClass  = {quant-ph},
  doi           = {10.48550/arXiv.2609.40351},
  url           = {https://arxiv.org/abs/2609.40351}
}

\appendix
\section{Learning \texorpdfstring{$\twoDQAC$}{2D-QAC} functions via matrix product states}
\label{sec:mps-learning}


We bound the columnwise matrix product state (MPS) bond dimension of
outputs of narrow $\twoDQAC$ circuits. Applied to the phase encoding in
\eqref{eq:phase-uncompute}, this places the computed Boolean functions
in a tensor-train classification model. The learning guarantee then
follows from the capacity bound of Khavari and Rabusseau
\cite{khavari2021pseudodimension} and Hanneke's PAC learning theorem
\cite{hanneke2016optimal}.

Regard each column of $[w]\times[h]$ as one subsystem of dimension
$2^w$, ordered from left to right.  For general local dimension $b$, an
open-boundary MPS has the form
\[
 \ket\psi
 =\sum_{s_1,\ldots,s_h=1}^{b}
       A_1[s_1]A_2[s_2]\cdots A_h[s_h]\ket{s_1,\ldots,s_h}.
\]
Here $A_j[s_j]$ is an $r_{j-1}\times r_j$ matrix, the boundary
dimensions are $r_0=r_h=1$, and the bond dimension is at most $\chi$
when all $r_j\le\chi$. We use the standard characterization: the
minimal bond dimension at each cut equals the Schmidt rank across that
cut \cite{vidal2003slightly}.

\begin{theorem}
\label{thm:mps-representation}
The output of a depth-$d$, width-$w$ $\twoDQAC$ circuit on a product input has a columnwise MPS representation of bond dimension at
most $2^{dw}$.  If the circuit decides $f$ exactly from product
ancillas, its phase state $\ket{\chi_f}\ket{0^a}$ from
\eqref{eq:phase-uncompute} has bond dimension at most $2^{2dw}$.
\end{theorem}

\begin{proof}
Fix a cut between consecutive columns.  A generalized controlled-$Z$
gate crossing this cut is a product reflection
\[
 I-2P_{\mathrm{left}}\otimes P_{\mathrm{right}},
\]
where each projector acts on the gate's support on its side of the cut,
with identities on the remaining qubits.  This is a sum of two product
operators across the cut.  Since gates in a multi-qubit layer have
disjoint supports, at most one horizontal gate per row crosses the cut,
giving at most $w$ crossing gates.  All other gates act on one side.
Expanding the crossing reflections therefore expresses the layer as a
sum of at most $2^w$ product operators across the cut.

An operator that is a sum of $R$ product operators sends a state of
Schmidt rank $S$ to a sum of at most $RS$ product vectors.  The input
has Schmidt rank one, and one-qubit gates do not change that rank.
Iterating over the $d$ multi-qubit layers bounds the Schmidt rank across
every column cut by $2^{dw}$. The Schmidt-rank characterization
\cite{vidal2003slightly} gives the claimed MPS representation.
When $h=1$, there is no internal
bond and the assertion holds with bond dimension one.

For the second assertion, absorb initialization of the product
ancillas into local gates.  The phase-and-uncompute circuit in
\eqref{eq:phase-uncompute} synthesizes $\ket{\chi_f}\ket{0^a}$ from a product
state at depth $2d$ on the same grid.  It applies even when the
original exact decision circuit leaves input-dependent garbage.
The first assertion now gives the bound $2^{2dw}$.
\end{proof}

We recall two learning bounds in the forms needed here.

\begin{lemma}[{\cite[Theorem~2 and Corollary~6]{khavari2021pseudodimension}}]
\label{lem:tensor-train-vc}
Let $h,b,r\ge1$ be integers. The class of Boolean functions
$s\mapsto\mathbf1\{B(s)<0\}$ on $[b]^h$, where $B$ ranges over
real tensor trains with $h$ sites and bond dimension at most $r$,
has VC dimension $O(h\mkern1.5mu b\mkern1.5mu r^2\log(2h))$.
\end{lemma}

\begin{theorem}[{\cite[Theorem~2]{hanneke2016optimal}}]
\label{thm:realizable-pac-learning}
Let $\mathcal H$ be a Boolean hypothesis class on a finite domain
with VC dimension $v$. For every distribution $D$, target
$f\in\mathcal H$, and $\varepsilon,\delta\in(0,1)$, there is a learner
using $O((v+\log(1/\delta))/\varepsilon)$ independent examples
$(x,f(x))$, with $x\sim D$, that returns $\widetilde f$ satisfying
$\Pr_{x\sim D}[\widetilde f(x)\ne f(x)]\le\varepsilon$ with probability
at least $1-\delta$. The hypothesis need not belong to $\mathcal H$;
the result does not bound the running time of a consistent base learner.
\end{theorem}

To apply this representation to learning, fix the parameters $h,w,d$
and the locations of the $n$ input qubits and $a=wh-n$ ancillas,
all known to the learner.  The ancillas begin in a fixed product state,
which we take to be $\ket{0^a}$ by one-qubit gates.
Let $\mathcal F_{h,w,d}$ be the Boolean functions decided exactly by
circuits on this layout with depth at most $d$.  The learner receives
independent classical examples $(x,f(x))$, with $x$ uniform on
$\bits^n$, and may return a hypothesis outside $\mathcal F_{h,w,d}$.

Applying \cref{lem:tensor-train-vc,thm:realizable-pac-learning} gives
the following sample bound.

\begin{theorem}
\label{thm:mps-learning}
For every $\varepsilon,\delta\in(0,1)$, the class
$\mathcal F_{h,w,d}$ can be learned from
\[
 O\Biggl(
 \frac{h\,2^{(4d+1)w}\,\log(2h)
       +\log(1/\delta)}{\varepsilon}
 \Biggr)
\]
uniform labeled examples.  With probability at least $1-\delta$, the
returned hypothesis $\widetilde f$ satisfies
\[
 \Pr_x[\widetilde f(x)\ne f(x)]\le\varepsilon.
\]
\end{theorem}

\begin{proof}
Set $b=2^w$ and $r=2^{2dw}$. By \cref{thm:mps-representation},
the coefficient tensor $A$ of $\ket{\chi_f}\ket{0^a}$ has a complex
MPS representation of bond dimension at most $r$. The real tensor
$B=2^{n/2}\operatorname{Re}((1+i)A)$ satisfies
\[
 B(x,0^a)=\operatorname{Re}((1+i)i^{f(x)})=(-1)^{f(x)}.
\]
Representing each complex bond register by its real and imaginary
parts gives a real MPS for $B$ of bond dimension at most $2r$.
Thus $f$ belongs to the class of sign classifiers obtained by
evaluating such real tensor trains at $(x,0^a)$, with a negative
score interpreted as bit one and a fixed choice at zero.

Applying \cref{lem:tensor-train-vc} with bond dimension $2r$ gives
$v=O(h\,b\,r^2\log(2h))$ for their sign classifiers; restricting the
inputs to the known layout cannot increase this dimension.
The learner from \cref{thm:realizable-pac-learning}
therefore uses $O((v+\log(1/\delta))/\varepsilon)$ examples.
Substituting $b$ and $r$ proves the stated bound.

The consistent base classifiers required by that learner can be found
by enumerating rational MPS tensors. Such tensors realize the target
labels: the scores $B(x,0^a)$ are $\pm1$ on a finite domain, so
sufficiently small rational perturbations of the cores preserve all
their signs. This ensures a finite search for each consistent base
classifier, without asserting a time bound.
\end{proof}

\section{\texorpdfstring{$\Parity$}{PARITY} equivalence in \texorpdfstring{$\twoDQAC$}{2D-QAC}}
\label{sec:cat-fanout-exactification}

The connection between parity, fanout, and cat states dates back to
Moore's study of shallow quantum circuits \cite{moore1999quantum} and
was developed further by Green, Homer, Moore, and Pollett \cite{moor}.
In the unrestricted $\QAC$ model, synthesizing the parity and fanout
unitaries and synthesizing cat states are equivalent up to constant-depth
reductions. These reductions connect Boolean computation with the
creation of quantum correlations. Rosenthal later extended this approach
to approximate synthesis using nekomata, a generalization of cat states
\cite{rosenthal:LIPIcs.ITCS.2021.32}.
With the task definitions from \cref{sec:parity-cat-definitions} in
place, we establish the conversions in the grid geometry and their
approximation guarantees.

Our main result shows that approximate clean parity synthesis can be
made exact with only a constant-factor increase in depth.

\begin{theorem}
\label{thm:exact-amplification}
For fixed $p>1/2$, let $C$ be a depth-$d$
$\left[w,h\right]$-$\twoDQAC$ circuit that $p$-approximately implements clean
parity on $n$ controls, with its logical qubits filling a boundary
column. Then, exact clean parity on $n$ controls has a depth-$O\left(d\right)$
$\left[w+O\left(1\right),h+O\left(1\right)\right]$-$\twoDQAC$ implementation.
\end{theorem}

The new logical inputs may occupy added boundary positions. If $C$
uses $a$ ancillas, the exact circuit uses $O\left(n+a+1\right)$ ancillas;
the conversion preserves polynomial resources.

\Cref{fig:parity-fanout-conversions} summarizes the conversions and
exact amplification proved in this appendix. We restate the conversions
from \cref{sec:parity-cat-definitions} before proving them.

\begin{figure}[htbp]
\centering
\begin{tikzpicture}[
  task/.style={draw=black!45,rounded corners=3pt,fill=black!2,
    text width=4.7cm,minimum height=1cm,align=center,
    font=\small,inner sep=5pt},
  basic/.style={line width=0.85pt,draw=black!75,>=Stealth},
  overhead/.style={basic,dashed,draw=blue!65!black},
  reference/.style={font=\footnotesize,fill=white,inner sep=3pt},
  every path/.style={rounded corners=5pt}
]
  \node[task] (parity) at (0,0) {Clean parity};
  \node[task] (fanout) at (8.6,0) {Clean fanout};
  \node[task] (exact) at (0,-2.5) {Exact nekomata synthesis};
  \node[task] (approxparity) at (8.6,-5) {Approximate\\clean parity};
  \node[task] (approximate) at (0,-5) {Approximate\\nekomata synthesis};
  \node[task] (cat) at (8.6,-2.5) {Cat synthesis};

  \draw[basic,<->] (parity.east)--node[reference,above]
    {\cref{prop:parity-fanout-reductions}}(fanout.west);
  \draw[basic,->] (fanout.south)
    --node[reference,left]{\cref{prop:parity-fanout-reductions}}(cat.north);
  \draw[overhead,->] (cat.south)
    --node[reference,left]{\cite[Theorem~3.1]{rosenthal:LIPIcs.ITCS.2021.32}}(approxparity.north);
  \draw[basic,->] (approxparity.west)--node[reference,above]
    {\cref{prop:parity-fanout-reductions}}(approximate.east);
  \draw[overhead,->] (approximate.north)
    --node[reference,right]{\cref{lem:endpoint-amplification}}(exact.south);
  \draw[overhead,->] (exact.north)
    --node[reference,right]{\cref{prop:parity-fanout-reductions}}(parity.south);
\end{tikzpicture}
\caption{Conversions and exact amplification in
\cref{sec:cat-fanout-exactification}. Here $n$ counts logical qubits for
unitaries and target qubits for state synthesis, $d$ is depth, and
$w,h$ are the numbers of rows and columns. Solid arrows preserve
$(n,d,w,h)$. Dashed arrows map these parameters to
$(n+O\left(1\right),O\left(d\right),w+O\left(1\right),h+O\left(1\right))$ under mild conditions;
see the cited results for their precise statements.}
\label{fig:parity-fanout-conversions}
\end{figure}

\parityfanoutreductions*

The parity--fanout conversion and this cat synthesis also preserve
approximate clean squared fidelity.

\begin{proof}
The parity--fanout equivalence and the reduction to cat synthesis are
standard \cite[Proposition~3.1]{moor}; see also
\cite[Theorem~3.1]{rosenthal:LIPIcs.ITCS.2021.32}.
They use only Hadamard conjugation and product-state initialization,
so they preserve the layout, multi-qubit depth, and clean squared
fidelity guarantee.

For exact decision computation, use the standard
compute--copy--uncompute procedure. When the output is on the
boundary, its value can be copied to an adjacent new accumulator by a
CNOT before reversing the computation. This gives a clean implementation
with depth $O\left(d\right)$ and one additional boundary row or column.

The nekomata-to-parity conversion is the exact case of
\cite[Theorem~3.1]{rosenthal:LIPIcs.ITCS.2021.32}.
We describe how to realize its gates in the required geometry.
Let $A$ be the synthesis circuit on $n$ targets and $a$ ancillas.
The reduction applies $A$, couples input $x_j$ to target $j$ by a
controlled-$Z$ gate, applies $A^\dagger$, and computes the OR of the
$n+a$ working qubits into the parity accumulator. It then reverses the
working-register computation to clean the ancillas.

Place the input qubits in a fresh column beside the boundary targets.
The controlled-$Z$ gates are then disjoint adjacent pairs. For the OR,
append a column of zero flags, compute each row's OR into its flag in
parallel, compute the OR of the flags into the accumulator in an
additional boundary row, and uncompute the flags. Each OR is a
generalized Toffoli conjugated by single-qubit $X$ gates, with support
in one row or column. The reduction therefore adds only a constant
number of rows and columns and has depth $O\left(d\right)$, since it calls $A$
and $A^\dagger$ only a constant number of times.
\end{proof}

We now restate \cref{lem:endpoint-amplification} and prove it using
the amplitude-amplification mechanism of
\cite[Theorem~7 and Corollary~8]{grier2026mathsfqac0containsmathsftc0with}.
We take the synthesis depth $d$ to be at least one.

Recall from \cref{def:state-synthesis} that $p$-approximate nekomata
synthesis requires squared fidelity at least $p$ with some
nekomata $\ket\nu$, including its auxiliary register.

\endpointamplification*

\begin{proof}
Write the original circuit as $A$, acting on $n$ target qubits and $a$ ancillas.  Projecting
its output onto the two constant target strings gives
\[
 A\ket{0^{n+a}}
 =a_0\ket{0^n}\ket{\eta_0}
  +a_1\ket{1^n}\ket{\eta_1}+\ket\omega.
\]
Here $0\le a_j\le1$, the $\eta_j$ are normalized, and
$\ket\omega$ has no component with either constant target string.
The approximate nekomata guarantee and Cauchy--Schwarz give
$a_0+a_1\ge\sqrt{2p}$, so both $a_j$ are at least
$\lambda=\sqrt{2p}-1>0$.
Take $k=\ceil{\pi/(4\lambda)}$ and $\theta=\pi/(4k+2)$.
Then $\sin\theta\le\theta<\lambda$ and $k=O(1/\lambda)$.
Set $u_j=\sin\theta/(\sqrt2a_j)$.  Then,
$|u_0|^2+|u_1|^2\le\sin^2\theta/\lambda^2\le1$, so a two-qubit
unitary $Q$ can satisfy
\[
 Q\ket{00}=u_0\ket{00}+u_1\ket{11}
       +\sqrt{1-|u_0|^2-|u_1|^2}\ket{01}.
\]
An arbitrary two-qubit unitary has a constant-size decomposition into
one-qubit gates and CNOTs.  The choice of $Q$ uses the exact endpoint
weights, as permitted by our nonuniform model with arbitrary one-qubit
gates.  Append these two qubits to the target
register.  The amplitudes of the two enlarged constant strings are
both $\sin\theta/\sqrt2$, with their original garbage
states.  All remaining components have a nonconstant target string.

Append a flag which is zero exactly on the two constant strings.
Initialize it to one and flip it using the all-zero and all-one
target tests.  Denote the resulting state-synthesis circuit by $V$.
Its good branch is
\[
 \ket G=\frac{\ket{0^{n+2}}\ket{\eta_0}
                    +\ket{1^{n+2}}\ket{\eta_1}}{\sqrt2}\ket0.
\]
Thus, $V\ket0=\sin\theta\ket G+\cos\theta\ket B$ for a normalized
bad branch $\ket B$ with flag one.  The two branches are orthogonal.
Here $\ket0$ at a circuit input denotes its complete zero register.
Let $S$ multiply flag-zero states by $-1$, and let
$R=V(2\proj0-I)V^\dagger$ be the reflection about $V\ket0$.
On the plane spanned by $G,B$, direct multiplication gives
\[
 RS\bigl(\sin\alpha\ket G+\cos\alpha\ket B\bigr)
 =\sin(\alpha+2\theta)\ket G
    +\cos(\alpha+2\theta)\ket B.
\]
Thus, $(RS)^kV\ket0=\ket G$, because $(2k+1)\theta=\pi/2$.
In unrestricted geometry the reflection about zero is one generalized
controlled phase, up to local gates and a global phase.  The circuit
uses $2k+1$ calls to $V$ or $V^\dagger$ and reuses its registers.
It needs $n+a+O\left(1\right)$ qubits and $O\left((s+1)/\lambda\right)$ multi-qubit gates
when $A$ has $s$ gates.

For the geometric assertion, add the two targets and the flag in a
constant number of rows extending the target column.  Both endpoint
tests are vertical interval Toffolis; $Q$ acts on adjacent qubits.
Reflection about the all-zero working register is also available in
constant geometric depth.  Compute row ORs into a new column, apply a
phase conditioned on all flags being zero, and uncompute those ORs.
The row flags are temporary ancillas restored after each reflection,
so they need not themselves be included in the reflected register.
All inserted unused qubits start in zero and remain fixed.  Consequently,
the reflection tests exactly the working register, and the same
two-dimensional rotation proves exactness without increasing width or
length by more than a constant.
\end{proof}

\begin{proof}[Proof of \cref{thm:exact-amplification}]
By Hadamard conjugation and a product-state initialization,
\cref{prop:parity-fanout-reductions} gives a state $\ket\psi$ satisfying
\[
 \left|\bigl(\bra{\mathrm{cat}_{n+1}}\otimes\bra{0^a}\bigr)
 \ket\psi\right|^2\ge p.
\]
This is a $p$-approximate $(n+1)$-nekomata. Since $p>1/2$ is fixed,
\cref{lem:endpoint-amplification} gives an exact $(n+3)$-nekomata
in depth $O\left(d\right)$.
Keep $n$ targets and include the other three in the garbage register.
\Cref{prop:parity-fanout-reductions} then gives exact clean parity.
The working register becomes workspace for the new inputs, using
$O\left(n+a+1\right)$ ancillas in total.

To preserve the two-dimensional layout, amplify in the original target column,
then choose a contiguous set of $n$ targets adjacent to a fresh input
column.  The discarded targets remain inside the workspace
and are included in its zero test.  The row-OR implementations in the
lemma and proposition establish the claimed geometric depth and space.
Without geometric restrictions, the same argument exactifies a clean
$\QAC$ parity circuit with depth $O\left(d\right)$ and $O\left(n+a+1\right)$ ancillas.
\end{proof}

Finally, we restate \cref{lem:geometric-fanout-simulation}, which supplies
the geometric simulation used in \cref{thm:boolean-synthesis}.

\geometricfanoutsimulation*

\begin{proof}
Hadamard conjugation of the base gadget in
\cref{thm:exact-parity-construction}(d) gives exact clean fanout
in constant depth. Its height is linear in the number of logical
qubits, and its length is at most $\exp(O(r\log(r+2)))$.
All logical ports can lie in the left boundary column: the parity
controls already do, and an extra boundary row lets us swap the
accumulator to its required position and back. Ordinary Toffoli
gates also have column layouts, with linear height and constant length.

Place the original $n+a$ qubits in the top row, one per column.
For each layer, stack its gate gadgets in disjoint horizontal bands
to the right, with their left boundaries aligned. Route each
participating qubit vertically to its assigned port row and then
horizontally to the gadget. The vertical swaps use distinct columns,
and the horizontal swaps use distinct rows. Their paths avoid the
workspaces and have fixed intermediate qubits, so the endpoint swaps
from \cref{thm:two-dimensional-simulation} take constant depth.

Run the gadgets in parallel and reverse the routing. Their clean
workspace guarantees give the original layer on arbitrary inputs.
Since its gates have disjoint supports, their heights sum to
$O(n+a)$; adjoining the data columns gives the claimed rectangle.
Reuse this rectangle for all layers, setting and restoring fixed
ancilla values by one-qubit gates.
\end{proof}

\end{document}